\documentclass[11pt]{article}
\usepackage[utf8]{inputenc}

\usepackage[bookmarks,colorlinks,breaklinks]{hyperref}
\hypersetup{urlcolor=blue, colorlinks=true, citecolor=green!50!black, linkcolor=blue}
\usepackage[letterpaper, left=1in, right=1in, top=0.9in, bottom=0.9in]{geometry}
\usepackage[T1]{fontenc}
\usepackage{bold-extra}
\usepackage[american]{babel}
\usepackage{etoolbox}
\usepackage[normalem]{ulem}

\usepackage{graphicx}
\graphicspath{{images/}{../images/}}

\usepackage{amsmath, amssymb, cases}
\usepackage{theoremstyles}
\usepackage{thm-restate}
\usepackage{enumitem}
\usepackage{mdframed}
\usepackage{bbm}
\usepackage{bm}
\usepackage{mathtools}

\usepackage[ruled]{algorithm}
\usepackage[noend]{algpseudocode}

\usepackage{microtype}

\usepackage{mdframed}
\usepackage[dvipsnames]{xcolor}

\usepackage{tikz}
\usetikzlibrary{fit, arrows, shapes, matrix, chains, positioning, backgrounds, arrows.meta, tikzmark, decorations.pathreplacing, calc}

\tikzset{
bicolor/.style 2 args={
  dashed,dash pattern=on 5pt off 5pt,->,#1,
  postaction={draw,dashed,dash pattern=on 5pt off 5pt,->,#2,dash phase=5pt}
  },
}

\usepackage{mathrsfs}
\usepackage{makecell, multirow}
\setcellgapes{3pt}

\tikzset{me/.style={to path={
\pgfextra{% 
 \pgfmathsetmacro{\startf}{-(#1-1)/2}  
 \pgfmathsetmacro{\endf}{-\startf} 
 \pgfmathsetmacro{\stepf}{\startf+1}}
 \ifnum 1=#1 -- (\tikztotarget)  \else
     let \p{mid}=($(\tikztostart)!0.5!(\tikztotarget)$) 
         in
\foreach \i in {\startf,\stepf,...,\endf}
    {%
     (\tikztostart) .. controls ($ (\p{mid})!\i*6pt!90:(\tikztotarget) $) .. (\tikztotarget)
      }
      \fi   
     \tikztonodes
}}}

\usepackage{subfig}

\definecolor{color1}{RGB}{128,0,0}     %maroon
\definecolor{color2}{RGB}{170,110,40}   %brown
\definecolor{color3}{RGB}{128,128,0}    %olive
\definecolor{color4}{RGB}{0,128,128}     %teal
\definecolor{color5}{RGB}{0,0,128}       %navy
\definecolor{color6}{RGB}{0,0,0}        %black
\definecolor{color7}{RGB}{230,25,75}      %red
\definecolor{color8}{RGB}{245,130,48}  %orange
\definecolor{color9}{RGB}{255,225,25}  %yellow
\definecolor{color10}{RGB}{210,245,60}    %lime
\definecolor{color11}{RGB}{60,180,75}    %green
\definecolor{color12}{RGB}{70,240,240}   %cyan
\definecolor{color13}{RGB}{0,130,200}    %blue
\definecolor{color14}{RGB}{145,30,180}   %purple
\definecolor{color15}{RGB}{240,50,230}   %magenta
\definecolor{color16}{RGB}{128,128,128}  %grey
\definecolor{color17}{RGB}{250,190,212}  %pink
\definecolor{color18}{RGB}{215,215,180}  %apricot
\definecolor{color19}{RGB}{255,250,200}  %beige
\definecolor{color20}{RGB}{170,255,195}  %mint
\definecolor{color21}{RGB}{220,190,255}  %lavender
\definecolor{color22}{RGB}{255,255,255}  %white

\definecolor{palecerulean}{rgb}{0.61, 0.77, 0.89}
\definecolor{pistachio}{rgb}{0.58, 0.77, 0.45}
\definecolor{bubblegum}{rgb}{0.99, 0.76, 0.8}
\definecolor{buublecerulean}{rgb}{0.80,0.77,0.82}
\definecolor{buublepistachio}{rgb}{0.79,0.77,0.63}
\definecolor{ceruleanpistachio}{rgb}{0.59,0.77,0.67}
\definecolor{airforceblue}{rgb}{0.36, 0.54, 0.66}
\definecolor{amaranth}{rgb}{0.9, 0.17, 0.31}

\title{Finding Diverse Minimum s-t Cuts\footnote{An earlier version of this work appeared at the 34th International Symposium on Algorithms and Computation (ISAAC 2023).}}
\author{Mark de Berg \thanks{Eindhoven University of Technology, Netherlands, \texttt{m.t.d.berg@tue.nl}} \and Andr\'es L\'opez Mart\'inez \thanks{Eindhoven University of Technology, Netherlands, \texttt{a.lopez.martinez@tue.nl}} \and Frits Spieksma \thanks{Eindhoven University of Technology, Netherlands, \texttt{f.c.r.spieksma@tue.nl}}}

\date{\today}

\begin{document}

\maketitle

\pagenumbering{arabic}

\begin{abstract}
Recently, many studies have been devoted to finding \textit{diverse} solutions in classical combinatorial problems, such as \textsc{Vertex Cover} (Baste et al., IJCAI'20), \textsc{Matching} (Fomin et al., ISAAC'20) and \textsc{Spanning Tree} (Hanaka et al., AAAI'21). Finding diverse solutions is important in settings where the user is not able to specify all criteria of a desired solution. Motivated by an application in the field of \textit{system identification}, we initiate the algorithmic study of $k$-\textsc{Diverse Minimum s-t Cuts} which, given a directed graph $G = (V, E)$, two specified vertices $s,t \in V$, and an integer $k > 0$, asks for a collection of $k$ minimum $s$-$t$ cuts in $G$ that has maximum diversity. We investigate the complexity of the problem for maximizing three diversity measures that can be applied to a collection of cuts: (i) the sum of all pairwise Hamming distances, (ii) the cardinality of the union of cuts in the collection, and (iii) the minimum pairwise Hamming distance. We prove that $k$-\textsc{Diverse Minimum s-t Cuts} can be solved in strongly polynomial time for diversity measures (i) and (ii) via \textit{submodular function minimization}. We obtain this result by establishing a connection between ordered collections of minimum $s$-$t$ cuts and the theory of distributive lattices. When restricted to finding only collections of mutually disjoint solutions, we provide a more practical algorithm that finds a maximum set of pairwise disjoint minimum $s$-$t$ cuts. For graphs with small minimum $s$-$t$ cut, it runs in the time of a single \textit{max-flow} computation. Our results stand in contrast to the problem of finding $k$ diverse \textit{global} minimum cuts---which is known to be NP-hard even for the disjoint case (Hanaka et al., AAAI'23)---and partially answer a long-standing open question of Wagner (Networks, 1990) about improving the complexity of finding disjoint collections of minimum $s$-$t$ cuts. Lastly, we show that $k$-\textsc{Diverse Minimum s-t Cuts} subject to diversity measure (iii) is NP-hard already for $k=3$ via a reduction from a constrained variant of the minimum vertex cover problem in bipartite graphs.

% \noindent \textbf{Keywords} MinCut, Diversity, Lattice Theory, Submodular Function Minimization\\

% \noindent \textbf{Funding}
% This research was supported by the European Union’s Horizon 2020 research and innovation programme under the Marie Skłodowska-Curie grant agreement no. 945045, and by the NWO Gravitation project NETWORKS under grant no. 024.002.003.\\

% \noindent \textbf{Acknowledgements}
% We thank Martin Frohn for bringing the theory of lattices to our attention, and for fruitful discussions on different stages of this work. 
\end{abstract}

\section{Introduction}
\label{sec.introduction}
The \textsc{Minimum $s$-$t$ Cut} problem is a classic combinatorial optimization problem. Given a directed graph $G = (V, E)$ and two special %(nonadjacent) 
vertices $s, t \in V$, the problem asks for a subset $S \subseteq E$ of minimum cardinality that separates vertices $s$ and $t$, meaning that removing these edges from $G$ ensures there is no path from $s$ to $t$. Such a set is called a \textit{minimum $s$-$t$ cut} or \textit{$s$-$t$ mincut}, and it need not be unique. This problem has been studied extensively and has numerous practical and theoretical applications. Moreover, it is known to be solvable in polynomial time. Several variants and generalizations of the problem have been studied; we mention the \textit{global minimum cut} problem and the problem of \textit{enumerating all minimum $s$-$t$ cuts} in a graph.
In this paper, we initiate the algorithmic study of computing \textit{diverse} minimum $s$-$t$ cuts. Concretely, we introduce the following optimization problem. 

\begin{extthm}[$k$-\textsc{Diverse Minimum s-t Cuts} (\textsc{$k$-DMC})] \label{problem:1}
Given are a directed graph $G = (V, E)$, vertices $s,t \in V$, and an integer $k > 0$. Let $\Gamma_G(s, t)$ be the set of minimum $s$-$t$ cuts in $G$, and let $U_k$ be the set of $k$-element multisets of $\Gamma_G(s, t)$. We want to find $C \in U_k$ such that $d(C) = \max_{S \in U_k} d(S)$, where $d : U_k \rightarrow \mathbb{N}$ is a measure of diversity.
\end{extthm}

Informally, given a directed graph $G$ along with two specified vertices $s$ and $t$, and an integer $k$, we are interested in finding a collection of $k$ $s$-$t$ mincuts in $G$ that are as different from each other as possible; that is, a collection having \textit{maximum diversity}. Notice that the problem is well defined even when there are less than $k$ $s$-$t$ mincuts in $G$, because we allow multisets in the solution. To formally capture the notion of diversity of a collection of sets, several measures have been proposed in literature \cite{zheng2007finding, baste2019fpt, hanaka2021finding, BASTE2022103644, hanaka2022framework}. 
In this work, we choose three natural and general measures as our notions of diversity. Given a collection $(X_1, X_2, \ldots, X_k)$ of subsets of a set $A$ (not necessarily distinct), we define 
\begin{align}
    & d_{\text{sum}}(X_1, \ldots, X_k) = \sum_{1\leq i < j \leq k} |X_i \triangle X_j|, \label{eq:1} \\
    & d_{\text{cov}}(X_1, \ldots, X_k) = \big\lvert \bigcup_{1 \leq i \leq k} X_i \big\rvert, \; \text{and} \label{eq:3} \\
    & d_{\text{min}}(X_1, \ldots, X_k) = \min_{1\leq i \leq j \leq k} |X_i \triangle X_j| \label{eq:2}
\end{align}
where $X_i \triangle X_j = (X_i \cup X_j) \setminus (X_i \cap X_j)$ is the \textit{symmetric difference} (or \textit{Hamming distance}) of $X_i$ and $X_j$. Throughout, we call function \eqref{eq:1} the \textit{pairwise-sum diversity}, function \eqref{eq:3} the \textit{coverage diversity}, and function \eqref{eq:2} the \textit{bottleneck diversity}. 
These measures are amongst the most broadly used in describing diversity among solutions in combinatorial problems.  

\paragraph{Motivation.}
We now briefly motivate why finding diverse minimum $s$-$t$ cuts in a graph can be of interest. In general, to solve a real-world problem, one typically formulates the problem as an instance of a computational problem and proceeds to find a solution with the help of an optimization algorithm. However, this is not always an easy task, and the abstraction to a mathematical formulation is usually just a simplification. From a theoretical perspective, an optimal solution to the simplified problem is as good as any other optimal solution, but due to the loss of information during the abstraction process, not every such solution is guaranteed to be adequate for practical usage. 
An illustrating example is the so-called \textit{synthesis problem} in the field of \textit{system identification}, where (under special conditions) the \textsc{Minimum $s$-$t$ Cut} problem can be used to determine an optimal placement of input and output signals in a physical system (modeled as a directed graph) to gather information about its behavior \cite{SHI2022110093}. 
An optimal placement obtained from the abstract model, however, is not always practically feasible due to omitted physical constraints of the system that would otherwise render the model unmanageable \cite{paulShi2021Personal}.

One way of dealing with this issue is to present all optimal solutions of the simplified model and let a \textit{user} choose between them based on external factors ignored by the mathematical model. Such an approach is useful when the number of optimal solutions is small, but in most cases (as in the \textsc{Minimum $s$-$t$ Cut} problem) the number of optimal solutions can be exponential in the input size, rendering the approach infeasible. Another approach is to present only a small number $k$ of optimal solutions, but one should be careful not to output solutions that are very \textit{similar} to each other, as a solution resembling a practically infeasible solution is likely to be practically infeasible as well. Thus, we would like to somehow obtain a small list of $k$ optimal, yet sufficiently ``diverse'' solutions from which a user can make a choice a posteriori. 

\paragraph{Our results.} 
We investigate the complexity of the following three variants of $k$-\textsc{Diverse Minimum s-t Cuts}: (i) \textsc{Sum} $k$-\textsc{Diverse Minimum s-t Cuts} (\textsc{Sum}-$k$-\textsc{DMC}), (ii) \textsc{Cover} $k$-\textsc{Diverse Minimum s-t Cuts} (\textsc{Cov}-$k$-\textsc{DMC}), and (iii) \textsc{Min} $k$-\textsc{Diverse Minimum s-t Cuts} (\textsc{Min}-$k$-\textsc{DMC}). These are the problems obtained when defining function $d$ in \textsc{$k$-DMC} as diversity measures \eqref{eq:1}, \eqref{eq:3}, and \eqref{eq:2}, respectively. For a graph $G$, we use $n$ to denote the number of nodes and $m$ to denote the number of edges.

Contrary to the hardness of finding diverse \textit{global} mincuts in a graph \cite{hanaka2022framework}, in Section \ref{sec:sfm} we show that both \textsc{Sum}-$k$-\textsc{DMC} and \textsc{Cov}-$k$-\textsc{DMC} can be solved in polynomial time. We show this via a reduction to the \textit{submodular function minimization} problem (SFM) on a \textit{lattice}, which is known to be solvable in strongly polynomial time when the lattice is \textit{distributive} \cite{grotschel2012geometric, iwata2001combinatorial, schrijver2000combinatorial}.

\begin{restatable}[]{theorem}{mainTheorem} \label{thm:1}
    \textsc{Sum-$k$-DMC} and \textsc{Cov}-$k$-\textsc{DMC} can be solved in strongly polynomial time.
\end{restatable}

At the core of this reduction is a generalization of an old result of Escalante \cite{escalante1972schnittverbande} establishing a connection between minimum $s$-$t$ cuts and distributive lattices. %As will be elaborated in Section \ref{sec:sfm}, 
We obtain our results by showing that the pairwise-sum and coverage diversity measures (reformulated as minimization objectives) are \textit{submodular} functions on the lattice $L^*$ defined by \textit{left-right ordered} collections of $s$-$t$ mincuts, and that this lattice is in fact \textit{distributive}. Using the currently fastest algorithm for SFM by Jiang \cite{jiang2021minimizing}, together with an appropriate representation of the lattice $L^*$, we can obtain an algorithm solving these problems in $O(k^5n^5)$ time. 

In Section \ref{sec.disjointDMC}, we obtain better time bounds for the special case of finding collections of $s$-$t$ mincuts that are pairwise disjoint. Similar to \textsc{SUM-$k$-DMC} and \textsc{COV-$k$-DMC}, our approach exploits the partial order structure of $s$-$t$ mincuts. We use this to efficiently solve the following optimization problem, which we call $k$-\textsc{Disjoint Minimum $s$-$t$ Cuts}: given a graph $G = (V, E)$, vertices $s,t \in V$, and an integer $k \leq k_\mathrm{max}$, find $k$ pairwise disjoint $s$-$t$ mincuts in $G$. Here, $k_\mathrm{max}$ denotes the maximum number of disjoint $s$-$t$ mincuts in $G$. Our algorithm is significantly simpler than the previous best algorithm by Wagner \cite{wagner1990disjoint}, which uses a poly-logarithmic number of calls to any \textit{min-cost flow} algorithm. Our algorithm takes $O(F(m, n) + m\lambda)$ time, where $F(m, n)$ is the time required by a unit-capacity \textit{max-flow} computation, and $\lambda$ is the size of an $s$-$t$ mincut in the graph. By plugging in the running time of the current fastest deterministic max-flow algorithms of \cite{liu2020faster, kathuria2020potential}, we obtain the following time bounds. When $\lambda \leq m^{1/3 + o(1)}$, our algorithm improves upon the previous best runtime for this problem. 

\begin{theorem} \label{thm:3}
$k$-\textsc{Disjoint Minimum $s$-$t$ Cuts} can be solved in time $O(m^{4/3 + o(1)} + m \lambda)$. 
\end{theorem}

In Section \ref{sec.hardness}, we prove that the decision version of \textsc{Min-$k$-DMC} is already NP-hard when $k = 3$. The proof is split into three parts. First, we show that a variant of the \textit{constrained minimum vertex cover} problem on bipartite graphs (\textsc{Min-CVCB}) of Chen and Kanj \cite{chen2003constrained} is NP-hard. Then, we give a reduction from this problem to \textsc{2-Fixed 3-DMC}, a constrained version of \textsc{Min-$3$-DMC}. Finally, we provide a polynomial time reduction from \textsc{2-Fixed 3-DMC} to \textsc{Min-$3$-DMC}, which implies the hardness of the general problem.

\begin{theorem}
    The decision version of \textsc{Min-$k$-DMC} is NP-hard. 
\end{theorem}

\paragraph{Related Work.}
Many efforts have been devoted to finding diverse solutions in combinatorial problems. In their seminal paper \cite{kuo1993analyzing}, Kuo \textit{et al}. were the first to explore this problem from a complexity-theoretic perspective. They showed that the basic problem of maximizing a distance norm %\footnote{Keep in mind that the Hamming distance is not a norm.} 
over a set of elements is already NP-hard. Since then, the computational complexity of finding diverse solutions in many other combinatorial problems has been studied. For instance, diverse variants of \textsc{Vertex Cover}, \textsc{Matching} and \textsc{Hitting Set} have been shown to be NP-hard, even when considering simple diversity measures like the pairwise-sum of Hamming distances, or the minimum Hamming distance between sets. This has motivated the study of these and similar problems from the perspective of \textit{fixed-parameter tractable} (FPT) algorithms \cite{BASTE2022103644, fomin2020diverse, baste2019fpt}. 

Along the same line, Hanaka et al.~\cite{hanaka2022framework} and Gao et al. \cite{gao2022obtaining} recently developed frameworks to design approximation algorithms for diverse variants of combinatorial problems. On the positive side, diverse variants of other classic problems are known to be polynomially solvable when considering certain set-based diversity measures, such as \textsc{Spanning Tree} \cite{hanaka2021finding} and \textsc{Shortest Path} \cite{zheng2007finding}, but not much is known about graph partitioning problems in light of diversity. 

The problem of finding multiple minimum cuts has received considerable attention \cite{wagner1990disjoint, picard1980structure, hanaka2022framework}. Picard and Queyranne \cite{picard1980structure} initiated the study of finding all minimum $s$-$t$ cuts in a graph, showing that these can be enumerated efficiently. They observe that the closures of a naturally-defined poset over the vertices of the graph, correspond bijectively to minimum $s$-$t$ cuts. An earlier work of Escalante \cite{escalante1972schnittverbande} already introduced an equivalent poset for minimum $s$-$t$ cuts, but contrary to Picard and Queyranne, no algorithmic implications were given. Nonetheless, Escalante shows that the set of $s$-$t$ mincuts in a graph, together with this poset, defines a distributive lattice. Similar structural results for \textit{stable matchings} and \textit{circulations} have been shown to have algorithmic implications \cite{gusfield1989stable, khuller1993lattice}, but as far as we know, the lattice structure of $s$-$t$ mincuts has been seldomly exploited in the algorithmic literature.\footnote{Bonsma \cite{bonsma2010most} does make implicit use of the lattice structure of minimum $s$-$t$ cuts to investigate the complexity of finding \textit{most balanced minimum cuts} and \textit{partially ordered knapsack} problems, but does not make this connection to lattice theory explicit.}

Wagner \cite{wagner1990disjoint} studied the problem of finding $k$ pairwise-disjoint $s$-$t$ cuts of \textit{minimum total cost} in an edge-weighted graph.\footnote{Notice that when the input graph is unweighted and $k \leq k_{max}$, Wagner's problem reduces to $k$-\textsc{Disjoint Minimum $s$-$t$ Cuts}.} He showed that this problem can be solved in polynomial time by means of a reduction to a \textit{transshipment} problem; where he raised the question of whether improved complexity bounds were possible by further exploiting the structure of the problem, as opposed to using a general purpose \textit{min-cost flow} algorithm for solving the transshipment formulation. In sharp contrast, Hanaka et al. \cite{hanaka2022framework} recently established that the problem of finding $k$ pairwise-disjoint \textit{global} minimum cuts in a graph is NP-hard (for $k$ part of the input). We are not aware of any algorithm for minimum $s$-$t$ cuts that runs in polynomial time with theoretical guarantees on diversity.

\section{Preliminaries}\label{sec.preliminaries}

\subsection{Distributive Lattices} \label{sec.distribLattices}
% Good summary: https://www.kybernetika.cz/content/2019/2/252/paper.pdf
In this paper, we use properties of distributive lattices. Here we introduce basic concepts and results on posets and lattices while making an effort to minimize new terminology. For a more detailed introduction to lattice theory see e.g., \cite{birkhoff1937rings, davey2002introduction, gratzer2009lattice}. 

A \textit{partially ordered set (poset)} $P = (X, \preceq)$ is a ground set $X$ together with a binary relation $\preceq$ on $X$ that is reflexive, antisymmetric, and transitive. We use $\mathcal{D}(P)$ to denote the family of all ideals of $P$. When the binary operation $\preceq$ is clear from the context, we use the same notation for a poset and its ground set. Here, we consider the standard representation of a poset $P$ as a directed graph $G(P)$ containing a node for each element and edges from an element to its %successors 
predecessors. For a poset $P = (X, \preceq)$, an \textit{ideal} is a set $U \subseteq X$ where $u \in U$ implies that $v \in U$ for all $v \preceq u$. In terms of $G(P) = (V, E)$, a subset $W$ of $V$ is an ideal if and only if there is no outgoing edge from $W$. 

A \textit{lattice} is a poset $L = (X, \preceq)$ in which any two elements $x, y \in X$ have a (unique) greatest lower bound, or \textit{meet}, denoted by $x \wedge y$, as well as a (unique) least upper bound, or \textit{join}, denoted by $x \vee y$. We can uniquely identify $L$ by the tuple $(X, \vee, \wedge)$. A lattice $L'$ is a \textit{sublattice} of $L$ if $L' \subseteq L$ and $L'$ has the same meet and join operations as $L$. In this paper we only consider \textit{distributive lattices}, which are lattices whose meet and join operations satisfy distributivity; that is, $x \vee (y \wedge z) = (x \vee y) \wedge (x \vee z)$ and $x \wedge (y \vee z) = (x \wedge y) \vee (x \wedge z)$, for any $x,y,z \in L$. Note that a sublattice of a distributive lattice is also distributive. 

Suppose we have a collection $L_1, \ldots, L_k$ of lattices $L_i = (X_i, \vee_i, \wedge_i)$ with $i \in [k]$.\footnote{Throughout, we use $[k]$ to denote the set $\{1,...,k\}$.} The \textit{(direct) product lattice} $L_1 \times \ldots \times L_k$ is a lattice with ground set $X = \{(x_1, \ldots, x_k) \, : \, x_i \in L_i\}$ and join $\vee$ and meet $\wedge$ operations acting component-wise; that is, $x \vee y = (x_1 \vee_1 y_1, \ldots, x_k \vee_k y_k)$ and $x \wedge y = (x_1 \wedge_1 y_1, \ldots, x_k \wedge_k y_k)$ for any $x, y \in X$. The lattice $L^k$ is the product lattice of $k$ copies of $L$, and is called the $k$th \textit{power} of $L$. If $L$ is a distributive lattice, then $L^k$ is also distributive. 

A crucial notion we will need is that of \textit{join-irreducibles}. An element $x$ of a lattice $L$ is called \textit{join-irreducible} if it cannot be expressed as the join of two elements $y, z \in L$ with $y, z \neq x$. In a lattice, any element is equal to the join of all join-irreducible elements lower than or equal to it. The set of join-irreducible elements of $L$ is denoted by $J(L)$. Note that $J(L)$ is a poset whose order is inherited from $L$. Due to Birkhoff's representation theorem---a fundamental tool for studying distributive lattices---every distributive lattice $L$ is isomorphic to the lattice $\mathcal{D}(J(L))$ of ideals of its poset of join-irreducibles, with union and intersection as join and meet operations. Hence, a distributive lattice $L$ can be uniquely recovered from its poset $J(L)$.

\begin{theorem}[Birkhoff's Representation Theorem \cite{birkhoff1937rings}]
Any distributive lattice $L$ can be represented as the poset of its join-irreducibles $J(L)$, with the order induced from $L$.
\end{theorem}

For a distributive lattice $L$, this implies that there is a \textit{compact representation} of $L$ as the directed graph $G(L)$ that characterizes its set of join-irreducibles. (The graph $G(L)$ is unique if we remove transitive edges.) This is useful when designing algorithms, as the size of $G(L)$ is $O(|J(L)|^2)$, while $L$ can have as many as $2^{|J(L)|}$ elements. See Figure \ref{fig:compactRepresentation} for an illustration. 

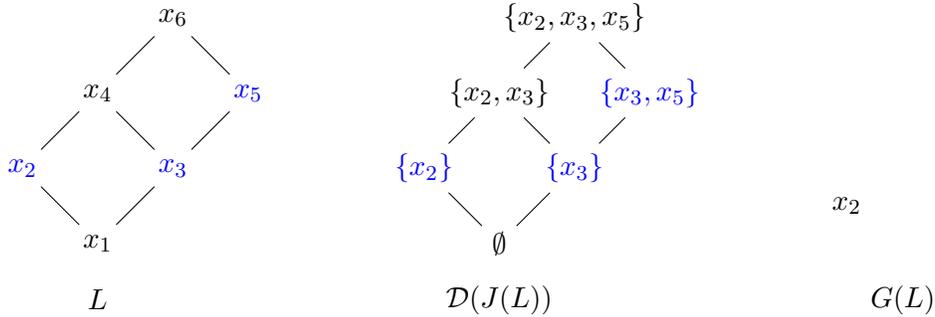
\begin{figure}
    \begin{center}
        \begin{tikzpicture}[scale=1.0, transform shape, baseline={(0,0)}]
    \tikzset{edge/.style = {->,> = latex}}
    
    \node[] (label) at (0, -0.75) {$L$};
    \node[] (n1) at (0, 0) {$x_1$};
    \node[] (n2) at (-1, 1) {\color{blue}{$x_2$}};
    \node[] (n3) at (1, 1) {\color{blue}{$x_3$}};
    \node[] (n4) at (0, 2) {$x_4$};
    \node[] (n5) at (2, 2) {\color{blue}{$x_5$}};
    \node[] (n6) at (1, 3) {$x_6$};
    
    \draw[] (n1) -- (n2);
    \draw[] (n1) -- (n3);
    \draw[] (n2) -- (n4);
    \draw[] (n3) -- (n4);
    \draw[] (n3) -- (n5);
    \draw[] (n4) -- (n6);
    \draw[] (n5) -- (n6);
    \end{tikzpicture}
    \hspace{1.25cm}
    \begin{tikzpicture}[scale=1.0, transform shape, baseline={(0,0)}]
    \tikzset{edge/.style = {->,> = latex}}
    
    \node[] (label) at (0, -0.75) {$\mathcal{D}(J(L))$};
    \node[] (n1) at (0, 0) {$\emptyset$};
    \node[] (n2) at (-1, 1) {\color{blue}{$\{x_2\}$}};
    \node[] (n3) at (1, 1) {\color{blue}{$\{x_3\}$}};
    \node[] (n4) at (0, 2) {$\{x_2, x_3\}$};
    \node[] (n5) at (2, 2) {\color{blue}{$\{x_3, x_5\}$}};
    \node[] (n6) at (1, 3) {$\{x_2, x_3, x_5\}$};
    
    \draw[] (n1) -- (n2);
    \draw[] (n1) -- (n3);
    \draw[] (n2) -- (n4);
    \draw[] (n3) -- (n4);
    \draw[] (n3) -- (n5);
    \draw[] (n4) -- (n6);
    \draw[] (n5) -- (n6);
    \end{tikzpicture}
    \hspace{1.25cm}
    \begin{tikzpicture}[scale=1.0, transform shape, baseline={(0,0)}]
    \tikzset{edge/.style = {->,> = latex}}
    
    \node[] (label) at (-0.75, -0.75) {$G(L)$};
    \node[] (n5) at (0, 2) {\color{black}{$x_5$}};
    \node[] (n2) at (-1.5, 0.5) {\color{black}{$x_2$}};
    \node[] (n3) at (0, 0.5) {\color{black}{$x_3$}};
    
    \draw[edge] (n5) -- (n3);
    \end{tikzpicture}
    \end{center}
    \vspace*{-4mm}
    \caption{Example of Birkhoff's representation theorem for distributive lattices. The left is a distributive lattice $L$, the middle is the isomorphic lattice $\mathcal{D}(J(L))$ of ideals of join-irreducibles of $L$, and the right shows the directed graph $G(L)$ representing $L$. The join irreducible elements of $L$ and $\mathcal{D}(J(L))$ are highlighted in blue.}
    \label{fig:compactRepresentation}
\end{figure}

\subsection{Submodular Function Minimization} \label{sec.sfmPreliminaries}
Let $f$ be a real-valued function on a lattice $L = (X, \preceq)$. We say that $f$ is \textit{submodular} on $L$ if 
\begin{equation} \label{eq:submodular}
    f(x \wedge y) + f(x \vee y) \leq f(x) + f(y), \quad \text{for all } x,y \in X\text{.}
\end{equation}
If $-f$ is submodular on $L$, then we say that $f$ is \textit{supermodular} in $L$ and just \textit{modular} if $f$ satisfies \eqref{eq:submodular} with equality. The \textit{submodular function minimization} problem (SFM) on lattices is, given a submodular function $f$ on $L$, to find an element $x \in L$ such that $f(x)$ is minimum. An important fact that we use in our work is that the sum of submodular functions is also submodular. Also, note 
that minimizing $f$ is equivalent to maximizing $-f$. 

Consider the special case of a lattice whose ground set $X \subseteq 2^U$ is a family of subsets of a set $U$, and meet and join are intersection and union of sets, respectively. It is known that any function $f$ satisfying \eqref{eq:submodular} on such a lattice can be minimized in polynomial time in $|U|$ \cite{grotschel2012geometric, iwata2001combinatorial, schrijver2000combinatorial}. This holds when assuming that for any $Y \subseteq U$, the value of $f(Y)$ is given by an \textit{evaluation oracle} that also runs in polynomial time in $|U|$. The current fastest algorithm for SFM on sets runs in $O(|U|^3 T_\mathrm{EO})$ time \cite{jiang2021minimizing}, where $T_\mathrm{EO}$ is the time required for one call to the evaluation oracle. 

Due to Birkhoff's theorem, the seemingly more general case of SFM on distributive lattices can be reduced to SFM on sets (see e.g. \cite[Sec. 3.1]{bolandnazar2015note} for details). 
%
% as follows.\footnote{For a more detailed description of the reduction from SFM on a distributive lattice to SFM on sets, we refer the reader to \cite[Sec. 3.1]{bolandnazar2015note}.} First, recall that every distributive lattice $L$ can be seen as the poset $\mathcal{D}(J(L))$ of ideals of its join-irreducibles, with union and intersection of ideals as join and meet operations, respectively. 
% Then, one can always construct an analogue function $\hat{f}$ on $\mathcal{D}(J(L))$ of the original function $f$ on $L$ in the following way. For the ideal $A \in \mathcal{D}(J(L))$ corresponding to the set of join-irreducibles lower than or equal to $a \in L$, simply set $\hat{f}(A) = f(a)$. %It then follows that 
% Then, the submodularity of $f$ on $L$ gets translated to the submodularity of $\hat{f}$ on $\mathcal{D}(J(L))$. Moreover, provided that union and intersection of sets can be computed in polynomial time, computing $\hat{f}$ is polynomially no harder than computing $f$. 
Hence, any polynomial-time algorithm for SFM on sets can be used to minimize a submodular function $f$ defined on a distributive lattice $L$. 
%by minimizing the analogue function $\hat{f}$ on $\mathcal{D}(J(L))$. 
An important remark is that the running time now depends on the size of the set $J(L)$ of join-irreducibles.

\begin{theorem}[{\cite[Note 10.15]{murota2003} and \cite[Thm.1]{markowsky2001overview}}] \label{theorem:sfmDistributiveLattices}
For any distributive lattice $L$, given by its poset of join-irreducibles $J(L)$, a submodular function $f: L \rightarrow \mathbb{R}$ can be minimized in polynomial time in $|J(L)|$, provided a polynomial time evaluation oracle for $f$.
\end{theorem}

\subsection{Minimum Cuts}
Throughout this paper, we restrict our discussion to directed graphs. All our results can be extended to undirected graphs by means of well-known transformations. Likewise, we deal only with edge-cuts, although our approach can be adapted to vertex-cuts as well.

Let $G$ be a directed graph with vertex set $V(G)$ and edge set $E(G)$. As usual, we define $n \coloneqq |V(G)|$ and $m \coloneqq |E(G)|$. 
Given a \textit{source} $s \in V(G)$ and \textit{target} $t \in V(G)$ in $G$, we call a subset $X \subset E(G)$ an \textit{$s$-$t$ cut} if the removal of $X$ from the graph ensures that no path from $s$ to $t$ exists in $G \setminus X$. The size of a cut is the total number of edges it contains. If an $s$-$t$ cut in $G$ has smallest size $\lambda(G)$, we call it a minimum $s$-$t$ cut, or an $s$-$t$ mincut. Note that such a cut need not be unique (in fact, there can be exponentially many). To denote the set of all $s$-$t$ mincuts of $G$, we use the notation $\Gamma_G(s, t)$. 

A (directed) path starting in a vertex $u$ and ending in a vertex $v$ is called a $u$-$v$ path. By Menger's theorem, the cardinality of a minimum $s$-$t$ cut in $G$ is equal to the maximum number of internally edge-disjoint $s$-$t$ paths in the graph. Let $\mathcal{P}_{s, t}(G)$ denote a maximum-sized set of edge-disjoint paths from $s$ to $t$ in $G$. Any minimum $s$-$t$ cut in $G$ contains exactly one edge from each path in $\mathcal{P}_{s, t}(G)$.

For two distinct edges (resp. vertices) $x$ and $y$ in a $u$-$v$ path $p$, we say that $x$ is a path-predecessor of $y$ in $p$ and write $x \prec_p y$ if the path $p$ meets $x$ before $y$. We use this notation indistinctly for edges and vertices. It is easily seen that the relation $\prec_p$ extends uniquely to a non-strict partial order. We denote this partial order by $x \preceq_p y$. Consider now any subset $W \subseteq \Gamma_G(s, t)$ of $s$-$t$ mincuts in $G$, and let %$E(W)$ denote the set of edges used in any of these cuts. 
let $E(W) = \bigcup_{X \in W} X$. Two crucial notions in this work are those of \textit{leftmost} and \textit{rightmost} $s$-$t$ mincuts. The \textit{leftmost} $s$-$t$ mincut in $W$ consists of the set of edges $S_{\text{min}}(W) \subseteq E(W)$ such that, for every path $p \in \mathcal{P}(s, t)$, there is no edge $e \in E(W)$ satisfying $e \prec_p f$ for any $f \in S_{\text{min}}(W)$. Similarly for the \textit{rightmost} $s$-$t$ mincut $S_{\text{max}}(W) \subseteq E(W)$. %Note that $S_\mathrm{min}(W)$ is not necessarily one of the cuts in $W$. It can be shown that 
Note that both $S_{\text{min}}(W)$ and $S_{\text{max}}(W)$ are also $s$-$t$ mincuts in $G$ (see proof of Claim \ref{claim:leftmostRightmost} in the appendix). When $W$ consists of the entire set of $s$-$t$ mincuts in $G$, we denote these extremal cuts by $S_\mathrm{min}(G)$ and $S_\mathrm{max}(G)$. 

On the set of $s$-$t$ cuts (not necessarily minimum), the following predecessor-successor relation defines a partial order: an $s$-$t$ cut $X$ is a predecessor of another $s$-$t$ cut $Y$, denoted by $X \leq Y$, if every path from $s$ to $t$ in $G$ meets an edge of $X$ at or before an edge of $Y$. 
It is known that the set of $s$-$t$ mincuts together with relation $\leq$ defines a distributive lattice $L$ \cite{escalante1972schnittverbande}. Moreover, a compact representation $G(L)$ can be constructed from a maximum flow in linear time \cite{picard1980structure}. These two facts play a crucial role in the proof of our main result in the next section.

\section{A Polynomial Time Algorithm for \textsc{SUM-k-DMC} and \textsc{COV-k-DMC}}\label{sec:sfm}
This section is devoted to proving Theorem \ref{thm:1} by reducing \textsc{SUM-$k$-DMC} and \textsc{COV-$k$-DMC} to SMF on distributive lattices. %Before giving the actual reduction, however, we need one additional step; which is to 
First, we show that the domain of solutions of \textsc{SUM-$k$-DMC} and \textsc{COV-$k$-DMC} can be restricted to the set of $k$-tuples that satisfy a particular order, as opposed to the set of $k$-sized multisets of $s$-$t$ mincuts (see Corollary \ref{corollary:solutionSpace} below). 
The reason for doing so is that the structure provided by the former set can be exploited to assess the ``modularity'' of 
the pairwise-sum and coverage objectives. We begin by introducing the notions of \textit{left-right order} and \textit{edge multiplicity}, which are needed throughout the section.

Consider a graph $G$ with specified $s, t \in V(G)$, and let $U^k$ be the set of all $k$-tuples over $\Gamma_G(s, t)$. An element $C \in U^k$ is a (ordered) collection or sequence $[X_1, \ldots, X_k]$ of cuts $X_i \in \Gamma_G(s, t)$, where $i$ runs over the index set $\{1, \ldots, k\}$. We say that $C$ 
is in \textit{left-right order} if $X_i \leq X_j$ for all $i < j$. Let us denote by $U_{\mathrm{lr}}^k \subseteq U^k$ the set of all $k$-tuples over $\Gamma_G(s, t)$ that are in left-right order. Then, for any two $C_1, C_2 \in U_{\mathrm{lr}}^k$, with $C_1 = [X_1, \ldots, X_k]$, $C_2 = [Y_1, \ldots, Y_k]$, we say that $C_1$ is a \textit{predecessor} of $C_2$ (and $C_2$ a \textit{successor} of $C_1$) if $X_i \leq Y_i$ for all $i \in [k]$, and denote this by $C_1 \preceq C_2$. Now, consider again a collection $C \in U^k$. The set of edges $\bigcup_{X \in C} X$ %$\bigcup\{E(X) : X \in C\}$ 
is denoted by $E(C)$. We define the \textit{multiplicity} of an edge $e \in E(G)$ with respect to $C$ as the number of cuts in $C$ that contain $e$ and denote it by $\mu_e(C)$. We say that an edge $e \in E(C)$ is a \textit{shared edge} if $\mu_e(C) \geq 2$. The set of shared edges in $C$ is denoted by $E_{\mathrm{shr}}(C)$. We make the following proposition, whose proof is deferred to Appendix \ref{appendix.two}.

\begin{restatable}{proposition}{propMultCons} \label{proposition:1}
For every $C \in U^k$ there exists $\hat{C} \in U_{\mathrm{lr}}^k$ such that $\mu_e(C) = \mu_{e}(\hat{C})$ for all $e \in E(G)$. 
\end{restatable}

In other words, given a $k$-tuple of $s$-$t$ mincuts, there always exists a $k$-tuple on the same set of edges that is in left-right order; each edge occurring with the same multiplicity. Consider now the pairwise-sum and the coverage diversity measures first introduced in Section \ref{sec.introduction}. We can rewrite them directly in terms of the multiplicity of shared edges as 
\begin{align}
    & d_{\text{sum}}(C) = 2 \left[\lambda(G) \binom{k}{2} - \sum_{e \in E_{\mathrm{shr}}(C)} \binom{\mu_e(C)}{2} \right], \quad \text{and} \label{eq:dsumMultiplicity} \\
    & d_{\text{cov}}(C) = k \lambda(G) - \sum_{e \in E_{\mathrm{shr}}(C)} \left( \mu_e(C) - 1 \right), \label{eq:dcovMultiplicity}
\end{align}
where terms outside the summations are constant terms. Then, combining eq. 
%Notice that the terms outside the summations are constant terms. Equation \eqref{eq:dsumMultiplicity} follows from the fact that we count a shared edge once per pair of cuts that contain it---and there are $\mu_e(C)$ such cuts---while equation \eqref{eq:dcovMultiplicity} follows from removing doubly counted edges. From combining 
\eqref{eq:dsumMultiplicity} (resp. \eqref{eq:dcovMultiplicity}) with Proposition \ref{proposition:1}, we obtain the following corollary. (For simplicity, we state this only for the $d_{\text{sum}}$ diversity measure, but an analogous claim holds for the $d_{\text{cov}}$ measure.)

\begin{corollary} \label{corollary:solutionSpace}
Let $C \in U^k$ such that $d_{\text{sum}}(C) = \max_{S \in U^k} d_{\text{sum}}(S)$. Then there exists $C' \in U_{\mathrm{lr}}^k$ such that $d_{\text{sum}}(C') = d_{\text{sum}}(C)$. 
\end{corollary}

This corollary tells us that in order to solve \textsc{SUM-$k$-DMC} (resp. \textsc{COV-$k$-DMC}) we do not need to optimize over the set $U_k$ of $k$-element multisets of $\Gamma_G(s, t)$. Instead, we can look at the set $U_{\mathrm{lr}}^k \subseteq U^k$ of $k$-tuples that are in left-right order. %On the other hand, 
Moreover, it follows from Eqs. \eqref{eq:dsumMultiplicity} and \eqref{eq:dcovMultiplicity} that the problem of maximizing $d_{\text{sum}}(C)$ and $d_{\text{cov}}(C)$ is equivalent to that of minimizing 
\begin{align}
    & \hat{d}_\mathrm{sum}(C) = \sum_{e \in E_{\mathrm{shr}}(C)} \binom{\mu_e(C)}{2}, \quad \text{and} \label{eq:dsumMultiplicityMin} \\
    & \hat{d}_\mathrm{cov}(C) = \sum_{e \in E_{\mathrm{shr}}(C)} \left( \mu_e(C) - 1 \right), \label{eq:dcovMultiplicityMin}
\end{align}
respectively. 
In turn, the submodularity of \eqref{eq:dsumMultiplicityMin} (resp. \eqref{eq:dcovMultiplicityMin}) %$\hat{d}_{\text{sum}}(C)$ (resp. $\hat{d}_{\text{cov}}(C)$) 
implies the supermodularity of \eqref{eq:dsumMultiplicity} (resp. \eqref{eq:dcovMultiplicity}) %$d_{\text{sum}}(C)$ (resp. $d_{\text{cov}}(C)$) 
and vice versa. In the remaining of the section, we shall only focus on the minimization objectives $\hat{d}_{\text{sum}}$ and $\hat{d}_{\text{cov}}$. 

We are now ready to show that both \textsc{SUM-$k$-DMC} and \textsc{COV-$k$-DMC} can be reduced to SFM. We first show that the poset $L^* = (U_{\mathrm{lr}}^k, \preceq)$ is a distributive lattice (Section \ref{sec.distributivity}). Next we prove that the diversity measures $\hat{d}_{\text{sum}}$ and $\hat{d}_{\text{cov}}$ are submodular functions on $L^*$ (Section \ref{sec.submodularity}). Lastly, we show that there is a compact representation of the lattice $L^*$ and that it can be constructed in polynomial time, concluding with the proof of Theorem \ref{thm:1} (Section \ref{sec.compactRepresentation}). 

\subsection{Proof of Distributivity} \label{sec.distributivity}
We use the following result of Escalante \cite{escalante1972schnittverbande} (see also \cite{meyer1982lattices, halin1993lattices}). Recall that $\leq$ denotes the predecessor-successor relation between two $s$-$t$ mincuts. 

\begin{lemma}[\cite{escalante1972schnittverbande}] \label{lemma:escalante}
The set $\Gamma_G(s, t)$ of $s$-$t$ mincuts of $G$ together with the binary relation $\leq$ forms a distributive lattice $L$. For any two cuts $X, Y \in L$, the join and meet operations are given respectively by 
\begin{align*}
    X \vee Y = & \; S_{\text{max}}(X \cup Y), \quad \text{and}\\
    X \wedge Y = & \; S_{\text{min}}(X \cup Y).
\end{align*}
\end{lemma}

By the definition of product lattice, we can extend this result to the relation $\preceq$ on the set $U_{\mathrm{lr}}^k$ of $k$-tuples of $s$-$t$ mincuts that are in left-right order. 

\begin{lemma} \label{lemma:distributivity}
The set $U_{\mathrm{lr}}^k$, together with relation $\preceq$, defines a distributive lattice $L^*$. For any two elements $C_1 = [X_1, \ldots, X_k]$ and $C_2 = [Y_1, \ldots, Y_k]$ in $L^*$, the join and meet operations are given respectively by
\begin{align*}
    C_1 \vee C_2 = & [S_{\mathrm{max}}(X_1 \cup Y_1), \ldots, S_{\mathrm{max}}(X_k \cup Y_k)], \quad \text{and}\\
    C_1 \wedge C_2 = & [S_{\mathrm{min}}(X_1 \cup Y_1), \ldots, S_{\mathrm{min}}(X_k \cup Y_k)].
\end{align*}
\end{lemma}
\begin{proof} 
This follows directly from Lemma \ref{lemma:escalante} and the definition of product lattice (see Section \ref{sec.distribLattices}). Let $L^k = (U^k, \preceq)$ be the $k$th power of the lattice $L = (\Gamma_G(s, t), \leq)$ of minimum $s$-$t$ cuts, and let $L^* = (U_{\mathrm{lr}}^k, \preceq)$ with $U_{\mathrm{lr}}^k \subseteq U^k$ be the sublattice of left-right ordered $k$-tuples of minimum $s$-$t$ cuts. We know from Section \ref{sec.preliminaries} that since $L$ is distributive, then so is the power lattice $L^k$. Moreover, any sublattice of a distributive lattice is also distributive. Hence, it follows that the lattice $L^*$ is also distributive.
\end{proof}

\subsection{Proof of Submodularity} \label{sec.submodularity}
Now we prove that the functions $\hat{d}_{\text{sum}}$ and $\hat{d}_{\text{cov}}$ are submodular on the lattice $L^*$. 
We start with %three
two lemmas that establish useful properties of the multiplicity function $\mu_e(C)$ on $L^*$ (see the corresponding proofs in Appendix \ref{appendix.multiplicityModular} and \ref{appendix.multiplicityProperty}). 

\begin{lemma} \label{lemma:multiplicityModular}
The multiplicity function $\mu_e: U_{\mathrm{lr}}^k \rightarrow \mathbb{N}$ is modular on $L^*$.  
\end{lemma}
% \begin{proposition} \label{proposition.interval}
% For any $C = [X_1, \ldots, X_k]$ in $L^*$, the edge $e \in E(C)$ appears in every cut of a contiguous subsequence $C' = [X_i, \ldots, X_j]$ of $C$, $1 \leq i \leq j \leq k$, with size $|C'| = \mu_e(C)$.
% \end{proposition}
\begin{restatable}[]{lemma}{multiplicityProperty} \label{lemma:multiplicityProperty}
For any two $C_1, C_2 \in L^*$ and $e \in E(C_1) \cup E(C_2)$, it holds that $\max(\mu_e(C_1 \vee C_2), \mu_e(C_1 \wedge C_2)) \leq \max(\mu_e(C_1), \mu_e(C_2))$.
\end{restatable}

Lemma \ref{lemma:multiplicityProperty} plays an important role in the submodularity of $\hat{d}_{\text{sum}}$ and $\hat{d}_{\text{cov}}$. In contrast to Lemma \ref{lemma:multiplicityModular}, it does not hold on the $k$th power lattice of the distributive lattice $L$ of Lemma \ref{lemma:escalante}. 

%\subsubsection{Submodularity of \texorpdfstring{$\hat{d}_\mathrm{sum}$}{dsum}}
\paragraph{Submodularity of $\hat{d}_\mathrm{sum}$.} 
Recall the definition of $\hat{d}_\mathrm{sum}(C)$ in eq. \eqref{eq:dsumMultiplicityMin}, and let $B_e: U_{\mathrm{lr}}^k \rightarrow \mathbb{N}$ be the function defined by $B_e(C) = \binom{\mu_e(C)}{2}$. We can rewrite eq. \eqref{eq:dsumMultiplicityMin} as $\hat{d}_\mathrm{sum}(C) = \sum_{e \in E_{\mathrm{shr}}(C)} B_e(C)$. The following is a consequence of Lemmas \ref{lemma:multiplicityModular} and \ref{lemma:multiplicityProperty} (see proof in Appendix \ref{appendix.binomialSubmodular}).

\begin{claim} \label{claim:binomialSubmodular}
For any two $C_1, C_2 \in L^*$ and $e \in E(G)$, we have 
$B(C_1 \vee C_2) + B(C_1 \wedge C_2) \leq B(C_1) + B(C_2)$.
\end{claim}

In other words, the function $B_e(C)$ is submodular in the lattice $L^*$. 
Now, recall that the sum of submodular functions is also submodular. Then, taking the sum of $B_e(C)$ over all edges $e \in E(G)$ results in a submodular function. 
From here, notice that $B_e(C) = 0$ for unshared edges; that is, when $\mu_e(C) < 2$. This means that such edges do not contribute to the sum. It follows that, for any two $C_1, C_2 \in L^*$, we have
\begin{equation*}
    \sum_{e \in E_\mathrm{shr}(C_1 \vee C_2)}B_e(C_1 \vee C_2) + \sum_{e \in E_\mathrm{shr}(C_1 \wedge C_2)}B_e(C_1 \wedge C_2) \leq \sum_{e \in E_\mathrm{shr}(C_1)}B_e(C_1) + \sum_{e \in E_\mathrm{shr}(C_2)}B_e(C_2). 
\end{equation*}
Each sum in this inequality corresponds to the definition of $\hat{d}_{sum}$ applied to the arguments $C_1 \vee C_2$, $C_1 \wedge C_2$, $C_1$ and $C_2$, respectively. Hence, by definition of submodularity, we obtain our desired result. 

\begin{theorem} \label{theorem:submodularityTotalSum}
The function $\hat{d}_{\text{sum}}: U_{\mathrm{lr}}^k \rightarrow \mathbb{N}$ is submodular on the lattice $L^*$.
\end{theorem}

%\subsubsection{Submodularity of \texorpdfstring{$\hat{d}_\mathrm{cov}$}{dcov}}
\paragraph{Submodularity of $\hat{d}_\mathrm{cov}$.} Consider the function $F_e(C) : U_{\mathrm{lr}}^k \rightarrow \mathbb{N}$ defined by $F_e(C) = \mu_e(C)-1$. It is an immediate corollary of Lemma \ref{lemma:multiplicityModular} that $F_e(C)$ is modular in $L^*$. Then, the sum $\sum_e F_e(C)$ taken over all edges $e \in E(G)$ is still a modular function. Notice that only shared edges in $C$ contribute positively to the sum, while the contribution of unshared edges can be neutral or negative. We can split this sum into two parts: the sum over shared edges $e \in E_\mathrm{shr}(C)$, and the sum over $e \in E(G) \setminus E_\mathrm{shr}(C)$. The latter sum can be further simplified to $|E(C)| - |E(G)|$ by observing that only the edges $e \in E(G) \setminus E(C)$ make a (negative) contribution. Therefore, we can write
\begin{equation} \label{eq:sumCoverage}
     \sum\nolimits_{e \in E(G)} F_e(C) = \left( \sum\nolimits_{e \in E_\mathrm{shr}(C)} (\mu_e(C) - 1)\right) + |E(C)| - |E(G)|.
\end{equation}
We know $\sum_e F_e(C)$ to be a modular function on $L^*$, hence for any two $C_1, C_2 \in L^*$ we have
\begin{equation*}
    \sum_{e \in E(G)} F_e(C_1 \vee C_2) + \sum_{e \in E(G)} F_e(C_1 \wedge C_2) = \sum_{e \in E(G)} F_e(C_1) + \sum_{e \in E(G)} F_e(C_2),
\end{equation*}
which, by equation \eqref{eq:sumCoverage}, is equivalent to
\begin{align} \label{eq:coverageEquality}
    & \left( \sum_{e \in E_\mathrm{shr}(C_1 \vee C_2)} (\mu_e(C_1 \vee C_2) - 1)  + \sum_{e \in E_\mathrm{shr}(C_1 \wedge C_2)} (\mu_e(C_1 \wedge C_2) - 1) \right) + |E(C_1 \vee C_2)| + |E(C_1 \wedge C_2)| \notag \\ 
    & = \left( \sum_{e \in E_\mathrm{shr}(C_1)} (\mu_e(C_1) - 1) + \sum_{e \in E_\mathrm{shr}(C_2)} (\mu_e(C_2) - 1) \right) + |E(C_1)| + |E(C_2)|. 
\end{align}
Now, from Lemmas \ref{lemma:multiplicityModular} and \ref{lemma:multiplicityProperty}, we observe the following property (see proof in Appendix \ref{appendix.edgeSetSizesInequality}). 

\begin{claim} \label{claim:edgeSetSizesInequality}
For any two $C_1, C_2 \in L^*$ we have $|E(C_1 \vee C_2)| + |E(C_1 \wedge C_2)| \geq |E(C_1)| + |E(C_2)|$.
\end{claim}

Given Claim \ref{claim:edgeSetSizesInequality}, it is clear that to satisfy equality in equation \eqref{eq:coverageEquality} it must be that: 
% \begin{equation*} %\label{eq:coverageInequality}
%     \sum_{e \in E_\mathrm{shr}(C_1 \vee C_2)} (\mu_e(C_1 \vee C_2) - 1) + \sum_{e \in E_\mathrm{shr}(C_1 \wedge C_2)} (\mu_e(C_1 \wedge C_2) - 1) \leq \sum_{e \in E_\mathrm{shr}(C_1)} (\mu_e(C_1) - 1) + \sum_{e \in E_\mathrm{shr}(C_2)} (\mu_e(C_2) - 1),
% \end{equation*}
\begin{align*} %\label{eq:coverageInequality}
    \sum_{e \in E_\mathrm{shr}(C_1 \vee C_2)} (\mu_e(C_1 \vee C_2) - 1) & + \sum_{e \in E_\mathrm{shr}(C_1 \wedge C_2)} (\mu_e(C_1 \wedge C_2) - 1) \\ 
    & \leq \sum_{e \in E_\mathrm{shr}(C_1)} (\mu_e(C_1) - 1) + \sum_{e \in E_\mathrm{shr}(C_2)} (\mu_e(C_2) - 1),
\end{align*}
from which the submodularity of $\hat{d}_\mathrm{cov}$ immediately follows.

\begin{theorem} \label{theorem:submodularityCoverage}
The function $\hat{d}_\mathrm{cov}: U_{\mathrm{lr}}^k \rightarrow \mathbb{N}$ is submodular on the lattice $L^*$.
\end{theorem}

\subsection{Finding the Set of Join-Irreducibles} \label{sec.compactRepresentation}
We now turn to the final part of the reduction to SFM. By Lemma \ref{lemma:distributivity}, we know that the lattice $L^*$ of left-right ordered collections of $s$-$t$ mincuts is distributive. And it follows from Theorems \ref{theorem:submodularityTotalSum} and \ref{theorem:submodularityCoverage} that the objective functions $\hat{d}_\mathrm{sum}$ and $\hat{d}_\mathrm{cov}$ are submodular in $L^*$. As discussed in Section \ref{sec.sfmPreliminaries}, it only remains to find an appropriate (compact) representation of $L^*$ in the form of its poset of join-irreducibles $J(L^*)$. 

Recall the distributive lattice $L$ of $s$-$t$ mincuts of a graph $G$, defined in Lemma \ref{lemma:escalante}. The leftmost cut $S_\mathrm{min}(G)$ can be seen as the meet of all elements in $L$. In standard lattice notation, this smallest element is often denoted by $0_L := \bigvee_{x \in L} x$. We use the following result of Picard and Queyranne.% \cite{picard1980structure}. 

\begin{lemma}[see \cite{picard1980structure}] \label{lemma:picard}
    Let $L$ be the distributive lattice of $s$-$t$ mincuts in a graph $G$, there is a compact representation $G(L)$ of $L$ with the following properties: 
    \begin{enumerate}
        \item The vertex set is $J(L) \cup 0_L$,
        \item $|G(L)| \leq |V(G)|$,
        \item Given $G$ as input, $G(L)$ can be constructed in $F(n, m) + O(m)$ time. 
    \end{enumerate}
\end{lemma}

In other words, the set $J(L)$ is of size $O(n)$ and can be recovered from $G$ in the time of a single max-flow computation. Moreover, each element of $J(L)$ corresponds to an $s$-$t$ mincut in $G$. In view of this lemma, we obtain the following for the poset of join-irreducibles $J(L^*)$. %(see proof in Appendix \ref{appendix.joinIrreduciblesLStar}). 

\begin{lemma} \label{lemma:joinIrreduciblesLStar}
    The set of join-irreducibles of $L^*$ is of size $O(kn)$ and is given by
    \begin{center}  
        $J(L^*) = \bigcup_{i = 1}^k J_i$, where $J_i := \{(\underbrace{0_L, \ldots, 0_L}_{i-1 \text{ times}}, \underbrace{p, \ldots, p}_{k-i+1 \text{ times}})\, : \, p \in J(L)\}$.
    \end{center}
\end{lemma}
\begin{proof}
We know that for an element $x \in L^*$ such that $x \neq 0_L$, by definition of join-irreducible, $x \in J(L^*)$ if and only if $x$ has a single immediate predecessor in $L^*$. To prove our claim, we show that (i) the $k$-tuples $J_i$, with $1 \leq i \leq k$, are in $L^*$ and satisfy this property, and (ii) that no other tuple in $L^*$ satisfies it. 
    
    For (i), let $C(i, p)$ denote the $k$-tuple $(0_L, \ldots, 0_L, p, \ldots, p) \in J_i$, where the first $i-1$ entries contain $0_L$ and the remaining $k-i+1$ entries contain the element $p$, with $i \in [k]$ and $p \in J(L)$. It is clear that $C(i, p) \in L^*$since each entry in $C(i, p)$ is an $s$-$t$ mincut, and $0_L \leq X$ for any $X \in \Gamma_{s, t}(G)$. 
    Consider now the arbitrary element $p \in J(L)$, and let $q$ denote the immediate predecessor of $p$ in $J(L)$ (with $q = 0_L$ if $p$ has no predecessors). We claim that the $k$-tuple $Q(i, q, p) := (0_L, \ldots, 0_L, q, p, \ldots, p)$ obtained from $C(i, p)$ by replacing its $i$th entry with element $q$, is the unique immediate predecessor of $C(i, p)$. This follows because: (a) replacing any other entry of $C(i, p)$ with $q$ results in a tuple that violates the left-right order, (b) any other choice of $q$ either violates the order or has the tuple $Q(i, q, p)$ as a successor, and (c) replacing any subsequence of $p$s by $q$s in $C(i, p)$ has the same consequences as (b).\footnote{There is also the case where all $p$s are replaced by a $q$ such that $q > p$, but it is clear that no such tuple can be a predecessor of $C(i, p)$.} Since this holds for all $i\in [k]$ and arbitrary $p$, it follows that each tuple in $J(L^*)$ has a single immediate predecessor. 

    It remains to show (ii); that is, that there is no tuple in $L^* \setminus \bigcup_{i = 1}^k J_i$ which is also a join-irreducible of $L^*$. For the sake of contradiction, assume that such a tuple $T$ exists in $L^*$. There are two possibilities for $T$: (1) $T$ contains more than 2 elements from the set $J(L)$, and (2) $T$ contains no elements from $J(L)$. 
    
    Consider case (2) first, and let $\gamma$ be the $k$th entry in $T$. Since $\gamma \not\in J(L)$, then it has more than one immediate predecessor in $L$. Let $\alpha$ and $\beta$ be two such predecessors (notice that $\alpha$ and $\beta$ are incomparable). Then, we can construct two distinct tuples $T_1 \in L^*$ and $T_2 \in L^*$ from $T$ by replacing $\gamma$ by $\alpha$ and $\beta$, respectively. But $T_1$ and $T_2$ are both immediate predecessors of $T$ in $L^*$, which gives the necessary contradiction. 

    Case (1) follows a similar argument. Suppose $a, b, c \in J(L)$ are the last three entries in tuple $T$; where $a < b < c$. Let $p(c), p(b) \in J(L)$ be the immediate predecessors of elements $c$ and $b$, respectively. Notice that $a \leq p(b)$ and $b \leq p(c)$. Then, like before, we can construct two distinct tuples $T_1 \in L^*$ and $T_2 \in L^*$ from $T$ by replacing $c$ by $p(c)$ and $b$ by $p(b)$, respectively. It is clear that $T_1$ and $T_2$ are both immediate predecessors of $T$ in $L^*$, which once more results in a contradiction. 

    From (i) and (ii) above, we have thus shown that the set of join-irreducibles $J(L^*)$ is given by $\bigcup_{i = 1}^k J_i$. To conclude the proof, we look at the size of $J(L^*)$. First, observe that the index $i$ runs from $1$ to $k$. Also, by Lemma \ref{lemma:picard} we know that $|J(L)| = O(n)$. It then follows that $|J(L^*)| = O(kn)$.
\end{proof}

Given Lemma \ref{lemma:joinIrreduciblesLStar}, a compact representation of the lattice $L^*$ can be obtained as the directed graph $G(L^*)$ that characterizes its poset of join-irreducibles $J(L^*)$ in polynomial time (since $|J(L^*)|$ is polynomial). It is also clear that the functions $\hat{d}_\mathrm{sum}$ and $\hat{d}_\mathrm{cov}$ can be computed in polynomial time. Then, by Theorem \ref{theorem:sfmDistributiveLattices}, together with Theorems \ref{lemma:distributivity}, \ref{theorem:submodularityTotalSum} and \ref{theorem:submodularityCoverage}, the reduction to SFM is complete. 

\mainTheorem*

To give a precise running time bound, we can use Jiang's algorithm \cite{jiang2021minimizing} for minimizing a submodular function on sets. The total running time of our algorithm is $O(|U|^3 T_\mathrm{EO})$, where $|U| = O(k n)$ is the size of the ground set $J(L^*)$, and $T_\mathrm{EO} = O(k^2 n^2)$ is the time required to evaluate the analogue function on $\mathcal{D}(J(L^*))$ of the function $\hat{d}_\mathrm{sum}$ (resp. $\hat{d}_\mathrm{cov}$) on $L^*$. The graph representation of the poset $J(L^*)$ can be constructed within the same time bounds since $|G(L^*)| = O(k^2n^2)$. Thus, we get the following result (see Appendix \ref{appendix.runningTime} for a detailed derivation of the time bound.)

\begin{restatable}[]{theorem}{preciseruntime}
    \textsc{Sum-$k$-DMC} and \textsc{Cov}-$k$-\textsc{DMC} can be solved in $O(k^5n^5)$ time.
\end{restatable}

\section{A Simple Algorithm for Finding Disjoint Minimum s-t Cuts}\label{sec.disjointDMC}

In the previous section, we looked at the problem of finding the $k$ most diverse minimum $s$-$t$ cuts in a graph. Here, we consider a slightly different problem. Observe that for diversity measures $d_{\mathrm{sum}}$ and $d_{\mathrm{cov}}$, the maximum diversity is achieved when the elements of a collection are all pairwise disjoint. Thus, it is natural to ask for a maximum cardinality collection of $s$-$t$ mincuts that are pairwise disjoint; i.e., that are as diverse as possible. We call this problem \textsc{Maximum Disjoint Minimum $s$-$t$ Cuts} (or \textsc{Max-Disjoint MC} for short). 

\begin{extthm}[\textsc{Max-Disjoint MC}] \label{problem:MaxDisjointMC}
Given a graph $G = (V, E)$ and vertices $s,t \in V(G)$, find a set $S \subseteq \Gamma_G(s, t)$ such that $X \cap Y = \emptyset$ for all $X, Y \in S$, and $|S|$ is as large as possible. 
\end{extthm}

%Now, recall $k$-\textsc{Disjoint Minimum $s$-$t$ Cuts} from Section \ref{sec.introduction}. Observe that one can easily obtain a solution to this problem by simply returning any $k$-sized subset of cuts from a solution to \textsc{Max-Disjoint MC}. Hence, any algorithm for \textsc{Max-Disjoint MC} can be used to solve $k$-\textsc{Disjoint Minimum $s$-$t$ Cuts} within the same time bound. 
Observe that a solution to \textsc{Max-Disjoint MC} immediately yields a solution to $k$-\textsc{Disjoint Minimum $s$-$t$ Cuts}. In this section, we prove Theorem \ref{thm:3} by giving an algorithm for \textsc{Max-Disjoint MC} that runs in $O(F(m, n) + \lambda(G)m)$ time, where $F(m, n)$ is the time required by a max-flow computation. First, we look at a restricted case when the input graph can be decomposed into a collection of edge-disjoint $s$-$t$ paths and (possibly) some additional edges---we refer to such a graph as an \textit{$s$-$t$ path graph}---and devise an algorithm that handles such graphs. Then, we use this algorithm as a subroutine to obtain an algorithm that makes no assumption about the structure of the input graph. 

\subsection{When the input is an s-t path graph}

%We start with the formal definition of an \textit{$s$-$t$ path graph}. 

% \begin{definition}
% Let $H_{s,t}$ be a graph with designated vertices $s$ and $t$. We call $H_{s,t}$ an \emph{$s$-$t$ path graph} (or \emph{path graph} for short) if there is a collection $P_{s,t}$ of edge-disjoint $s$-$t$ paths such that $P_{s,t}$ covers all vertices in $V(H_{s,t})$; see Figure \ref{fig:pathGraph} for an illustration. The \emph{height} of $H_{s,t}$, denoted by $\lambda(H_{s,t})$, is the maximum number of edge-disjoint $s$-$t$ paths in the graph. For fixed $P_{s,t}$, we call the edges of $H_{s,t}$ in $P_{s,t}$ \textit{path edges} and edges of $H_{s,t}$ not in $P_{s,t}$ \textit{non-path edges}. Two vertices in $H_{s,t}$ are \textit{path neighbors} if they are joined by a path edge, and \textit{non-path neighbors} if they are joined (exclusively) by a non-path edge.
% \end{definition}

\begin{definition}
Let $H_{s,t}$ be a graph with designated vertices $s$ and $t$. We call $H_{s,t}$ an \emph{$s$-$t$ path graph} (or \emph{path graph} for short) if there is a collection $P$ of edge-disjoint $s$-$t$ paths such that $P$ covers all vertices in $V(H_{s,t})$. The \emph{height} of $H_{s,t}$, denoted by $\lambda(H_{s,t})$, is the maximum number of edge-disjoint $s$-$t$ paths in the graph. For fixed $P$, we call the edges of $H_{s,t}$ in $P$ \textit{path edges} and edges of $H_{s,t}$ not in $P$ \textit{non-path edges}. Two vertices in $H_{s,t}$ are \textit{path neighbors} if they are joined by a path edge, and \textit{non-path neighbors} if they are joined (exclusively) by a non-path edge. %When $s$ and $t$ are clear from the context, we simply refer to $H_{s,t}$ as a \textit{path graph}. 
See Figure \ref{fig:pathGraph} for an illustration.
\end{definition}

\begin{figure}
    \centering
    \begin{tikzpicture}[scale=0.9, transform shape]
    \tikzset{edge/.style = {->,> = latex}}
    
    % Nodes
    \node[label=left:{$s$},draw=black,circle] (s) at (-0.5, 0.5) {};
    \node[draw=black,circle] (p11) at (1, 2) {};
    \node[draw=black,circle] (p21) at (1, 1) {};
    \node[draw=black,circle] (p31) at (1, 0) {};
    \node[draw=black,circle] (p41) at (1, -1) {};
    
    \node[draw=black,circle] (p122) at (2, 1.5) {};
    \node[draw=black,circle] (p32) at (2, 0) {};
    \node[draw=black,circle] (p42) at (2, -1) {};
    
    \node[draw=black,circle] (p23) at (3, 2) {};
    \node[draw=black,circle] (p133) at (3, 1) {};
    \node[draw=black,circle] (p43) at (3, -1) {};
    
    \node[draw=black,circle] (p24) at (4, 2) {};
    \node[draw=black,circle] (p34) at (4, 0) {};
    \node[draw=black,circle] (p44) at (4, -1) {};
    
    \node[label=right:{$t$},draw=black,circle] (t) at (5.5, 0.5) {};
    
    % Edges
    \draw[edge] (s) -- node[above,yshift=0.25em,xshift=-0.25em] {$1$} (p11);
    \draw[edge] (s) -- node[above] {$2$} (p21);
    \draw[edge] (s) -- node[above] {$3$} (p31);
    \draw[edge] (s) -- node[below,yshift=-0.25em,xshift=-0.25em] {$4$} (p41);
    
    \draw[edge] (p11) -- node[above] {$1$} (p122);
    \draw[edge] (p21) -- node[below] {$2$} (p122);
    \draw[edge] (p31) -- node[above] {$3$} (p32);
    \draw[edge] (p41) -- node[below] {$4$} (p42);
    
    \draw[edge] (p122) -- node[above,xshift=-0.25em] {$2$} (p23);
    \draw[edge] (p122) -- node[below,xshift=-0.25em] {$1$} (p133);
    \draw[edge] (p32) -- node[below,xshift=0.25em] {$3$} (p133);
    \draw[edge] (p42) -- node[below] {$4$} (p43);
    
    \draw[edge] (p23) -- node[above] {$2$} (p24);
    \draw[edge] (p133) -- node[above,xshift=-0.25em] {$1$} (t);
    \draw[edge] (p133) -- node[below,xshift=-0.25em] {$3$} (p34);
    \draw[edge] (p43) -- node[below] {$4$} (p44);
    
    \draw[edge] (p24) -- node[above,xshift=0.25em] {$2$} (t);
    \draw[edge] (p34) -- node[below,xshift=-0.5em] {$3$} (t);
    \draw[edge] (p44) -- node[below,yshift=-0.25em,xshift=0.25em] {$4$} (t);
    
    % Gray edges
    \draw[edge,lightgray] (s) to [out=310,in=210] (p32.south);
    \draw[edge, lightgray] (p21) to (p32);
    \draw[edge, lightgray] (p23) to (t);
    \draw[edge, lightgray] (p31) to (p42);
    \draw[edge, lightgray] (p43) to (p34);
    \end{tikzpicture}
    \caption{Example of an $s$-$t$ path graph of height 4. Edges are labeled by integers corresponding to the path they belong to. Path edges are drawn in black and non-path edges in gray.}
    \label{fig:pathGraph}
\end{figure}
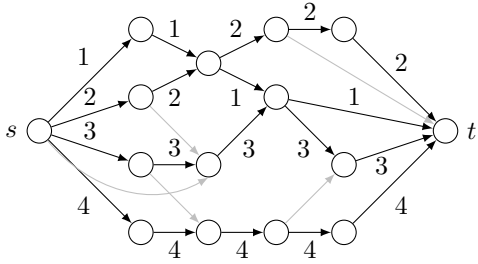

Two remarks are in order. The first is that, by Menger's theorem, the size of a minimum $s$-$t$ cut in an $s$-$t$ path graph coincides with its height. The second remark is that, from a graph $G$, one can easily obtain a path graph $H_{s,t}$ of height $\lambda(G)$ by finding a maximum-sized set $\mathcal{P}_{s,t}(G)$ of edge-disjoint $s$-$t$ paths in $G$ and letting $H_{s,t}$ be the induced subgraph of their union. Recall that, by Menger's theorem, a minimum $s$-$t$ cut in $G$ must contain exactly one edge from each path $p \in \mathcal{P}_{s,t}(G)$. Thus, every minimum $s$-$t$ cut of $G$ is in $H_{s,t}$. However, the reverse is not always true. In the above construction, there could be multiple new minimum $s$-$t$ cuts introduced in $H_{s, t}$ that arise from ignoring the reachability between %path vertices in $G$. 
vertices of $\mathcal{P}_{s,t}(G)$ in $G$. 
We will come back to this issue when discussing the general case in Section \ref{sec.generalCase}. 

\paragraph{The algorithm.}
%To reiterate, 
The goal in this subsection is to find a maximum cardinality collection $\hat{C}$ of pairwise disjoint $s$-$t$ mincuts in a path graph $H_{s, t}$. %Here, we provide an iterative algorithm to find such a collection with the additional property that $\hat{C}$ is in left-right order. 
We now explain the main ideas behind the algorithm. Without loss of generality, assume that the underlying set $\mathcal{P}_{s, t}(H_{s,t})$ of edge-disjoint $s$-$t$ paths that define $H_{s,t}$ is of maximum cardinality. %As usual, we denote this by $\mathcal{P}_{s, t}(H_{s,t})$. 
To simplify notation, we sometimes drop the argument and denote such set by $\mathcal{P}_{s, t}$. 

%here we denote such set $\mathcal{P}_{s, t}(H_{s,t})$ simply by $\mathcal{P}_{s, t}$. 

Let $X$ be an $s$-$t$ mincut in $H_{s,t}$, and suppose we are interested in finding an $s$-$t$ mincut $Y$ disjoint from $X$ such that $X < Y$. Consider any two edges $e = (u, u')$ and $f = (v, v')$ in $X$, and let $g = (w, w')$ be a path successor of $f$; that is $f \prec_p  g$ with $p \in \mathcal{P}_{s,t}$. If there is a non-path edge $h = (u', z)$ such that $w' \leq z$, we say that $h$ is \textit{crossing} with respect to $g$, and that $g$ is \textit{invalid} with respect to $X$ (see Figure \ref{fig:crossingInvalidEdges} for an illustration). The notions of crossing and invalid edges provide the means to identify the edges that cannot possibly be contained in $Y$. Let $E_\mathrm{inv}(X)$ denote the set of invalid edges with respect to $X$. We make the following observation.

\begin{figure}
    \centering
    \begin{tikzpicture}[scale=0.9, transform shape]
    \tikzset{edge/.style = {->,> = latex}}

    %\node[label=left:{$s$},draw=black,circle,fill=black] (s) at (-0.5, 0.5) {};
    \node[draw=black,circle,inner sep=2.5pt,minimum size=1pt] (u1) at (1,1) {};
    \node[draw=black,circle,inner sep=2.5pt,minimum size=1pt] (u2) at (3,1) {};

    \node[draw=black,circle,inner sep=2.5pt,minimum size=1pt] (v1) at (-1,0) {};
    \node[draw=black,circle,inner sep=2.5pt,minimum size=1pt] (v2) at (0,0) {};
    \node (dots1) at (1, 0) {$\dots$};
    \node[draw=black,circle,inner sep=2.5pt,minimum size=1pt] (w1) at (2,0) {};
    \node[draw=black,circle,inner sep=2.5pt,minimum size=1pt] (w2) at (3,0) {};
    \node (dots2) at (4, 0) {$\dots$};
    \node[draw=black,circle,inner sep=2.5pt,minimum size=1pt] (z) at (5,0) {};
        
    \draw[edge] (-1,1) node[left] {$\dots$} -- (u1);
    \draw[edge,blue] (u1) -- node[above] {$e$} (u2);
    %\draw[edge,blue] (u1) -- (u2);
    \draw[edge] (u2) -- (5,1) node[right] {$\dots$};

    \draw[edge] (-2,0) node[left] {$\dots$} -- (v1);
    \draw[edge,blue] (v1) -- node[above] {$f$} (v2);
    \draw[edge] (v2) -- (dots1);
    \draw[edge] (dots1) -- (w1);
    \draw[edge] (w1) -- node[above] {$g$} (w2);
    \draw[edge] (w2) -- (dots2);
    \draw[edge] (dots2) -- (z);
    \draw[edge] (z) -- (6,0) node[right] {$\dots$};

    \draw[edge,lightgray] (u2) -- node[above] {$h$} (z);
    \end{tikzpicture}
    \caption{Example illustrating the %Illustration of the 
    notions of crossing and invalid edges for an $s$-$t$ mincut $X$. Path edges are drawn in black and non-path edges in gray. Edges $e, f \in X$ are highlighted in blue. The edge $g$ is invalid with respect to $X$ since the edge $h$ is crossing with respect to it.}
    \label{fig:crossingInvalidEdges}
\end{figure}
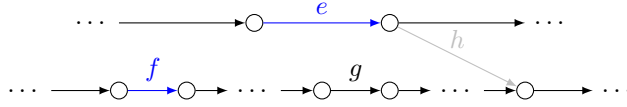 
 
\begin{observation} \label{observation:invalidBasic}
Let $Y > X$. Then $Y$ cannot contain an edge from $E_\mathrm{inv}(X)$.
\end{observation}
\begin{proof}
    For the sake of contradiction, suppose there exists an edge $g=(w,w')$ in $E_\mathrm{inv}(X)\cap Y$. Consider the path $p_1 \in \mathcal{P}_{s,t}$, and let $f$ be the predecessor of $g$ on $p_1$ that is in $X$. Since $g \in E_\mathrm{inv}(X)$, there is a crossing edge $h=(u',z)$ with respect to $g$. Let $p_2 \in \mathcal{P}_{s,t}$ be the path containing $u'$, and let $(u,u')$ be the edge of $p_2$ that is in $X$. Let $p_3$ be the s-t path that follows $p_2$ from $s$ to $u$, then follows the crossing edge $h$, and then continues along $p_1$ to $t$. Since $Y$ is an $s$-$t$ cut it must contain an edge from this path. Since $Y$ must contain exactly one edge from each path in $\mathcal{P}_{s,t}$, it cannot contain $h$. Moreover, $Y$ already contains edge $g$ from $p_1$. Then Y must contain an edge from the part of $p_2$ from $s$ to $u'$. But this contradicts that $Y>X$. 
\end{proof}

\begin{algorithm}
\caption[Caption for LOF]{Obtain a Maximum Set of Disjoint Minimum $s$-$t$ Cuts.} \label{algo.kdisjointMSP}
\vspace{.5em}
%\hspace*{\algorithmicindent} 
\textbf{Input:} Path graph $H_{s, t}$. \\
%\hspace*{\algorithmicindent} 
\textbf{Output:} A maximum set $\hat{C}$ of disjoint $s$-$t$ mincuts in $H_{s, t}$. 
\vspace{.5em}
\begin{algorithmic}[1]
\newcommand\NoDo{\renewcommand\algorithmicdo{}}
\newcommand\ReDo{\renewcommand\algorithmicdo{\textbf{do}}}
\algrenewcommand\algorithmiccomment[1]{\hfill {\color{blue} \(\triangleright\) #1}}

\State Initialize collection $\hat{C} \leftarrow \emptyset$ and set $M \leftarrow \{s\}$. 
\While{$t$ is unmarked} \Comment{Traverse the graph from left to right.}
    \While{$M$ is not empty} \Comment{Marking step.}
        \For{each vertex $v \in M$}
            \For{each path $p \in \mathcal{P}_{s,t}$} \Comment{Identify invalid edges.}
                \State Identify the rightmost neighbor $u \in p$ of $v$ reachable by a non-path edge. 
                \If{$u$ is unmarked%and is not the \textit{immediate successor} of $v$ in $p$
                }
                    \State Mark $u$ and all (unmarked) vertices that are path-predecessors of $u$.
                \EndIf
            \EndFor
        \EndFor
        \State Set $M$ to the set of newly marked vertices.
    \EndWhile
    \State $X \leftarrow \bigcup \{(x, y) \in \mathcal{P}_{s,t} \, : \, x \text{ is marked, } y \text{ is unmarked}\}$. 
    \Comment{Cut-finding step.}
    \State $\hat{C} \leftarrow \hat{C} \cup \{X\}$. 
    \For{each $(x, y) \in X$} \Comment{Mark the head node of cut edges.}
        \State Mark $y$.
    \EndFor
    \State $M \leftarrow \bigcup_{(x,y) \in X} y$. \Comment{Newly marked vertices.}
\EndWhile
\State Return $\hat{C}$.
\end{algorithmic}
\vspace{.5em}
\end{algorithm}

If we extend the definition of $E_\mathrm{inv}(X)$ to also include all the edges that are path predecessors of edges in $X$, we obtain that, for any $s$-$t$ path $p \in \mathcal{P}_{s,t}$, the set of invalid edges along $p$ is a prefix of the path. As a result, 
%
% If we extend the definition of $E_\mathrm{inv}(X)$ to also include all the edges that are path predecessors of edges in $X$, we immediately obtain the following key property. 
% \begin{observation} \label{observation:invalidEdges}
%     For any $s$-$t$ path $p \in \mathcal{P}_{s,t}$, the poset $E_{\mathrm{inv}}(X) \cap p$ with order relation given by path-distance from $s$ is an ideal of the set $E(p)$ of edges of $p$.  
% \end{observation}
%
% Observation \ref{observation:invalidEdges} implies that 
if we can identify the (extended) set $E_{\mathrm{inv}}(X)$, then we can restrict our search of cut $Y$ to only the set of valid edges $E_\mathrm{val}(X):= E(H_{s,t}) \setminus E_\mathrm{inv}(X)$. 
This motivates the following iterative algorithm for finding a pairwise disjoint collection of $s$-$t$ mincuts: (1) Find the leftmost $s$-$t$ mincut $X$ in $H_{s,t}$, (2) identify the set $E_{\mathrm{inv}}(X)$ and find the leftmost $s$-$t$ mincut $Y$ amongst $E_\mathrm{val}(X)$, (3) set $X = Y$ and repeat step (2) until $E_\mathrm{val}(X) \cap p = \emptyset$ for any one path $p \in \mathcal{P}_{s,t}$, and finally (4) output the union of identified cuts as the returned collection $\hat{C}$. Informally, notice that the $s$-$t$ mincut identified at iteration $i$ is a successor of the mincuts identified at iterations $j < i$. Hence, the returned collection will consist of left-right ordered and pairwise disjoint $s$-$t$ mincuts. Moreover, picking the leftmost cut at each iteration prevents the set of invalid edges from growing unnecessarily large, which allows for more iterations and thus, a larger set returned. Next, we give a more formal description of the algorithm, the details of which are presented in Algorithm \ref{algo.kdisjointMSP}.

% Notice that the $s$-$t$ mincut identified at iteration $i$ is a (strict) successor of the mincuts identified at iterations $j < i$. Hence, the returned collection will consist of left-right ordered and pairwise disjoint $s$-$t$ mincuts. Moreover, picking the leftmost cut at each iteration prevents the set of invalid edges from growing unnecessarily large, which allows for more iterations and thus, a larger set returned. Next, we give a more formal description of the algorithm, the details of which are presented in Algorithm \ref{algo.kdisjointMSP}.  %The details are presented in Algorithm \ref{algo.kdisjointMSP}. 

The algorithm works by traversing the graph from left to right in iterations while marking the vertices it visits. Initially, all vertices are unmarked, except for $s$. Each iteration consists of two parts: a \textit{marking} step, and a \textit{cut-finding} step. In the marking step (Lines 3-9), the algorithm identifies currently invalid edges by marking the non-path neighbors---and their path-predecessors---of currently marked vertices. (Observe that a path edge becomes invalid if both of its endpoints are marked.) In Algorithm \ref{algo.kdisjointMSP}, this is realized by a variable $M$ that keeps track of the vertices that have just been marked as a consequence of the marking of vertices previously present in $M$. 
In the cut-finding step (Lines 10-14), the algorithm then finds the leftmost minimum $s$-$t$ cut amongst valid path edges. Notice that, for each $s$-$t$ path in $\mathcal{P}_{s,t}$, removing its first valid edge prevents $s$ from reaching $t$ via that path. This means that our leftmost cut of interest is simply the set of all path edges that have exactly one of their endpoints marked. Following the identification of this cut, the step concludes by marking the head vertices of the identified cut edges. 
Finally, the algorithm terminates when the target vertex $t$ is visited and marked. See Figure \ref{fig:sweepingAlgorithm} for an example execution of the algorithm. 

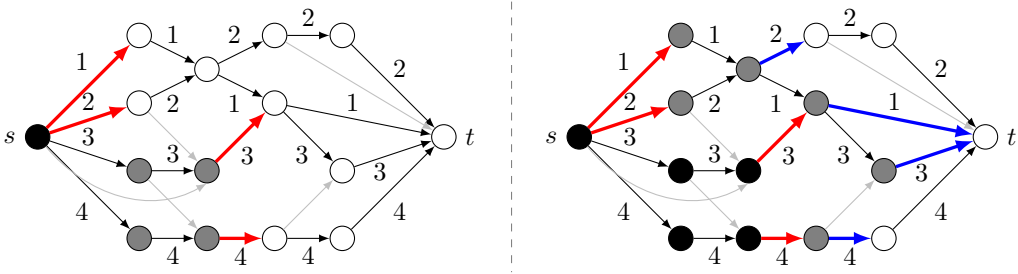
\begin{figure}
    \centering
    \begin{tikzpicture}[scale=0.9, transform shape]
    \tikzset{edge/.style = {->,> = latex}}
    
    %% First iteration
    % Nodes
    \node[label=left:{$s$},draw=black,circle,fill=black] (s) at (-0.5, 0.5) {};
    \node[draw=black,circle] (p11) at (1, 2) {};
    \node[draw=black,circle] (p21) at (1, 1) {};
    \node[draw=black,circle,fill=black!50] (p31) at (1, 0) {};
    \node[draw=black,circle,fill=black!50] (p41) at (1, -1) {};
    
    \node[draw=black,circle] (p122) at (2, 1.5) {};
    \node[draw=black,circle,fill=black!50] (p32) at (2, 0) {};
    \node[draw=black,circle,fill=black!50] (p42) at (2, -1) {};
    
    \node[draw=black,circle] (p23) at (3, 2) {};
    \node[draw=black,circle] (p133) at (3, 1) {};
    \node[draw=black,circle] (p43) at (3, -1) {};
    
    \node[draw=black,circle] (p24) at (4, 2) {};
    \node[draw=black,circle] (p34) at (4, 0) {};
    \node[draw=black,circle] (p44) at (4, -1) {};
    
    \node[label=right:{$t$},draw=black,circle] (t) at (5.5, 0.5) {};
    
    % Edges
    \draw[edge,red,line width=1.25pt] (s) -- node[above,yshift=0.25em,xshift=-0.25em] {\color{black}$1$} (p11);
    \draw[edge,red,line width=1.25pt] (s) -- node[above] {\color{black}$2$} (p21);
    \draw[edge] (s) -- node[above] {$3$} (p31);
    \draw[edge] (s) -- node[below,yshift=-0.25em,xshift=-0.25em] {$4$} (p41);
    
    \draw[edge] (p11) -- node[above] {$1$} (p122);
    \draw[edge] (p21) -- node[below] {$2$} (p122);
    \draw[edge] (p31) -- node[above] {$3$} (p32);
    \draw[edge] (p41) -- node[below] {$4$} (p42);
    
    \draw[edge] (p122) -- node[above,xshift=-0.25em] {$2$} (p23);
    \draw[edge] (p122) -- node[below,xshift=-0.25em] {$1$} (p133);
    \draw[edge,red,line width=1.25pt] (p32) -- node[below,xshift=0.25em] {\color{black}$3$} (p133);
    \draw[edge,red,line width=1.25pt] (p42) -- node[below] {\color{black}$4$} (p43);
    
    \draw[edge] (p23) -- node[above] {$2$} (p24);
    \draw[edge] (p133) -- node[above,xshift=-0.25em] {$1$} (t);
    \draw[edge] (p133) -- node[below,xshift=-0.25em] {$3$} (p34);
    \draw[edge] (p43) -- node[below] {$4$} (p44);
    
    \draw[edge] (p24) -- node[above,xshift=0.25em] {$2$} (t);
    \draw[edge] (p34) -- node[below,xshift=-0.5em] {$3$} (t);
    \draw[edge] (p44) -- node[below,yshift=-0.25em,xshift=0.25em] {$4$} (t);
    
    % Gray edges
    \draw[edge,lightgray] (s) to [out=310,in=210] (p32.south);
    \draw[edge, lightgray] (p21) to (p32);
    \draw[edge, lightgray] (p23) to (t);
    \draw[edge, lightgray] (p31) to (p42);
    \draw[edge, lightgray] (p43) to (p34);
    
    % Separation
    \coordinate (a) at (6.5,2.5);
    \coordinate (b) at (6.5,-1.5);
    
    \draw[dashed,gray] (a) -- (b);
    
    %% Second iteration
    % Nodes
    \node[label=left:{$s$},draw=black,circle,fill=black] (s) at (7.5, 0.5) {};
    \node[draw=black,circle,fill=black!50] (p11) at (9, 2) {};
    \node[draw=black,circle,fill=black!50] (p21) at (9, 1) {};
    \node[draw=black,circle,fill=black] (p31) at (9, 0) {};
    \node[draw=black,circle,fill=black] (p41) at (9, -1) {};
    
    \node[draw=black,circle,fill=black!50] (p122) at (10, 1.5) {};
    \node[draw=black,circle,fill=black] (p32) at (10, 0) {};
    \node[draw=black,circle,fill=black] (p42) at (10, -1) {};
    
    \node[draw=black,circle] (p23) at (11, 2) {};
    \node[draw=black,circle,fill=black!50] (p133) at (11, 1) {};
    \node[draw=black,circle,fill=black!50] (p43) at (11, -1) {};
    
    \node[draw=black,circle] (p24) at (12, 2) {};
    \node[draw=black,circle,fill=black!50] (p34) at (12, 0) {};
    \node[draw=black,circle] (p44) at (12, -1) {};
    
    \node[label=right:{$t$},draw=black,circle] (t) at (13.5, 0.5) {};
    
    % Edges
    \draw[edge,red,line width=1.25pt] (s) -- node[above,yshift=0.25em,xshift=-0.25em] {\color{black}$1$} (p11);
    \draw[edge,red,line width=1.25pt] (s) -- node[above] {\color{black}$2$} (p21);
    \draw[edge] (s) -- node[above] {$3$} (p31);
    \draw[edge] (s) -- node[below,yshift=-0.25em,xshift=-0.25em] {$4$} (p41);
    
    \draw[edge] (p11) -- node[above] {$1$} (p122);
    \draw[edge] (p21) -- node[below] {$2$} (p122);
    \draw[edge] (p31) -- node[above] {$3$} (p32);
    \draw[edge] (p41) -- node[below] {$4$} (p42);
    
    \draw[edge,blue,line width=1.25pt] (p122) -- node[above,xshift=-0.25em] {\color{black}$2$} (p23);
    \draw[edge] (p122) -- node[below,xshift=-0.25em] {$1$} (p133);
    \draw[edge,red,line width=1.25pt] (p32) -- node[below,xshift=0.25em] {\color{black}$3$} (p133);
    \draw[edge,red,line width=1.25pt] (p42) -- node[below] {\color{black}$4$} (p43);
    
    \draw[edge] (p23) -- node[above] {$2$} (p24);
    \draw[edge,blue,line width=1.25pt] (p133) -- node[above,xshift=-0.25em] {\color{black}$1$} (t);
    \draw[edge] (p133) -- node[below,xshift=-0.25em] {$3$} (p34);
    \draw[edge,blue,line width=1.25pt] (p43) -- node[below] {\color{black}$4$} (p44);
    
    \draw[edge] (p24) -- node[above,xshift=0.25em] {$2$} (t);
    \draw[edge,blue,line width=1.25pt] (p34) -- node[below,xshift=-0.5em] {\color{black}$3$} (t);
    \draw[edge] (p44) -- node[below,yshift=-0.25em,xshift=0.25em] {$4$} (t);
    
    % Gray edges
    \draw[edge,lightgray] (s) to [out=310,in=210] (p32.south);
    \draw[edge, lightgray] (p21) to (p32);
    \draw[edge, lightgray] (p23) to (t);
    \draw[edge, lightgray] (p31) to (p42);
    \draw[edge, lightgray] (p43) to (p34);
    
    \end{tikzpicture}
    \caption{Example illustrating the first two iterations of Algorithm \ref{algo.kdisjointMSP} on a path graph of height 4. The black- and gray-shaded vertices represent vertices marked at the previous and current iterations, respectively. The red edges correspond to the $s$-$t$ mincut found at the end of the first (left) iteration. Similarly, the blue edges correspond to the $s$-$t$ mincut found at the second (right) iteration.}
    \label{fig:sweepingAlgorithm}
\end{figure} 

We now make the following claim about the complexity of the algorithm, followed by an analysis of its correctness. 

\begin{claim} \label{claim:sweepingAlgoComplexity}
The complexity of Algorithm \ref{algo.kdisjointMSP} on an $m$-edge, $n$-vertex path graph is $O(m \log n)$. 
\end{claim} 
\begin{proof}
Let $H_{s,t}$ be our input path graph. First, notice that each vertex $v \in H_{s, t}$ is visited at most $deg(v)$ times by the algorithm. This follows from the fact that $v$ is only visited whenever one of three cases occurs: (i) $v$ is reachable by a marked vertex via a non-path edge (Line 6), (ii) $v$ is a predecessor of a marked vertex $u$ on a path $p \in \mathcal{P}_{s,t}$ (Line 8), or (iii) $v$ is the head node of an identified minimum $s$-$t$ cut (Line 12). 
We know that $v$ can be the endpoint of at most $deg(v) - 2$ non-path edges. Similarly, $v$ can be the endpoint of at most 2 path edges. Since a vertex cannot be reached again by a previously traversed edge, the remark follows. 

Now, observe that each time a vertex is visited, the algorithm performs only $O(1)$ work, except for the step in Line 6 where each currently marked vertex $v \in M$ must identify its rightmost neighbor on each path in $\mathcal{P}_{s,t}$. We can assume that each vertex $v \in H_{s,t}$ is equipped with a data structure $A_v$ that given a query path $p \in P(H_{s, t})$, can answer in $O(\log n)$ time which is the rightmost neighbor $u$ of $v$ in $p$.\footnote{If we are willing to forego worst-case complexity for amortized complexity, we can assume a data structure with constant insert and query complexity via hash tables.} Therefore, as the algorithm performs $O(\log n)$ work each time it visits a vertex $v \in H_{s,t}$, and it does so at most $deg(v)$ times, the claim would follow. It only remains to analyze the preprocessing time of equipping the graph with such data structures. 

We claim that the graph can be preprocessed in $O(m \log n)$ time as follows. Assume that each node $u \in H_{s,t}$ has two variables $path(u)$ and $pos(u)$ which store the path to which it belongs and its position in said path, respectively. 
First, for each vertex $v \in H_{s,t}$ we initialize an empty list $A_v$ of tuples of the form $(a, b)$. Then, for each neighbor $u$ of $v$, query the list $A_v$ for the tuple $(x, y)$ such that $x = path(u)$. If it exists and $pos(u) > y$, set $y = pos(u)$. If it does not exist, then create the tuple $(path(u), pos(u))$ and insert it in $A_v$ in sorted order (by path). Since $A_v$ can be of size at most $\lambda(H_{s, t})$, it is clear that querying and inserting can be implemented in $O(\log (H_{s, t}))$ time by binary search. Equipping each vertex with these lists then requires $O(deg(v) \cdot \log (H_{s. t}))$ time per vertex. Thus, the total preprocessing time is $O(m \log (\lambda(H_{s,t})))$, which can be simplified to $O(m \log n)$.
\end{proof}

\paragraph{Correctness of Algorithm \ref{algo.kdisjointMSP}.} 
We note an important property of collections of $s$-$t$ mincuts. %(see proof in Appendix \ref{appendix.leftmostGraph}). 
(We use $d(C)$ to denote any of $d_\mathrm{sum}(C)$ or $d_\mathrm{cov}(C)$.)

\begin{claim} \label{claim:leftmostGraph}
Let $C$ be a left-right ordered collection of minimum $s$-$t$ cuts in a graph $G$, the collection $\tilde{C}$ obtained by replacing $S_{\mathrm{min}}(\bigcup_{X \in C} X)$ (resp. $S_\mathrm{max}(\bigcup_{X \in C} X)$) with $S_\mathrm{min}(G)$ (resp. $S_\mathrm{max}(G)$) has cost $d(\tilde{C}) \leq d(C)$. 
\end{claim}
% \begin{proof}
% For simplicity, let us denote $S_{\mathrm{min}}(C) := S_{\mathrm{min}}(\bigcup_{X \in C} X)$. By definition, we know that no edge of $\bigcup_{X \in C} X$ lies to the left of $S_{min}(G)$. Then replacing $S_{\mathrm{min}}(C)$ with $S_\mathrm{min}(G)$ can only decrease the number of pairwise intersections previously present between $S_{\mathrm{min}}(C)$ and the cuts in $C \setminus S_{\mathrm{min}}(C)$. Notice that our measures of diversity only penalize edge intersections. Hence, the cost of collection $\tilde{C}$ cannot be greater than that of $C$.
% \end{proof}
\begin{proof}
We prove this only for $S_\mathrm{min}(\cdot)$ as the proof for $S_\mathrm{max}(\cdot)$ is analogous. 
For simplicity, let us denote $S_{\mathrm{min}}(C) := S_{\mathrm{min}}(\bigcup_{X \in C} X)$. By definition, we know that no edge of $\bigcup_{X \in C} X$ lies to the left of $S_{min}(G)$. Then replacing $S_{\mathrm{min}}(C)$ with $S_\mathrm{min}(G)$ can only decrease the number of pairwise intersections previously present between $S_{\mathrm{min}}(C)$ and the cuts in $C \setminus S_{\mathrm{min}}(C)$. Notice that our measures of diversity only penalize edge intersections. Hence, the cost of collection $\tilde{C}$ cannot be greater than that of $C$.
\end{proof}

Now, consider an arbitrary collection of $k$ edge-disjoint $s$-$t$ mincuts in a path graph $H_{s,t}$. Corollary \ref{corollary:solutionSpace} implies that there also exists a collection of $k$ edge-disjoint $s$-$t$ mincuts in $H_{s,t}$ that is in left-right order. In particular, this is true for a collection of maximum cardinality $k_{\textbf{max}}$. Together with Claim \ref{claim:leftmostGraph}, this means that there always exists a collection $\hat{C}$ of edge-disjoint $s$-$t$ mincuts in $H_{s, t}$ with the following properties: 
\begin{enumerate}[label=(\roman*),leftmargin=2\parindent]
    % \item $\hat{C}$ has size $k_{\textbf{max}}$, \label{property.I}
    % \item $\hat{C}$ is in left-right order, and  \label{property.II}
    % \item $\hat{C}$ contains the leftmost minimum $s$-$t$ cut of $H_{s, t}$. \label{property.III}
    \item $\hat{C}$ has size $k_{\textbf{max}}$, \label{property.I}
    \item $\hat{C}$ is in left-right order, \label{property.II}
    \item $\hat{C}$ contains the leftmost $s$-$t$ mincut of $H_{s, t}$, and  \label{property.III}
    %\item $\hat{C}$ contains the rightmost $s$-$t$ mincut of $H_{s, t}$.
    \item The set $S_\mathrm{max}(\hat{C}) \cap S_\mathrm{max}(H_{s,t})$ is not empty. \label{property.IV} 
\end{enumerate}
We devote the rest of the subsection to proving the following lemma, which serves to prove the correctness of Algorithm \ref{algo.kdisjointMSP}. 

\begin{lemma} \label{lemma.algorithmCorrectness}
Algorithm \ref{algo.kdisjointMSP} returns a collection of edge-disjoint minimum $s$-$t$ cuts that satisfies Properties \ref{property.I}--\ref{property.IV}.
\end{lemma}

Let $\hat{C}$ denote the solution returned by the algorithm. First, we show that $\hat{C}$ contains only disjoint cuts. This follows from the fact that a cut can only be found amongst valid edges at any given iteration, and once an edge has been included in a cut, it becomes invalid at every subsequent iteration. Similarly, Properties \ref{property.II} and \ref{property.III} are consequences of the notion of invalid edges. We start by proving the latter. Let $X_1$ denote the leftmost cut in $\hat{C}$. For the sake of contradiction, assume there is a minimum $s$-$t$ cut $Y$ such that $e \prec_p  f$. Here, $e \in Y$, $f \in X_1$ and \textit{w.l.o.g.} $p$ is an $s$-$t$ path from any arbitrary maximum collection of $s$-$t$ paths in $H_{s,t}$. For the algorithm to pick edge $f = (u, u')$ as part of $X_1$ it must be that vertex $u$ is marked and $u'$ is not. We know that the predecessors of marked vertices must also be marked. Hence we know that both endpoints of edge $e$ are marked. But by definition, this means that edge $e$ is invalid, and cannot be in a minimum $s$-$t$ cut. This gives us the necessary contradiction, and $X_1$ must be the leftmost cut in the graph. 

We continue with Property \ref{property.II}. This property follows from the fact that, at any given iteration, the posets of invalid path-edges on each path of $H_{s,t}$ are ideals of the set of path edges. This means that the edges in the cut found by the algorithm at iteration $i$ are all path predecessors of an edge in the cut found at iteration $i+1$. Carrying on with Property \ref{property.IV}, we prove that it follows from the fact that the algorithm terminates when the target node $t$ is marked. Suppose, for the sake of contradiction, that the cuts $S_\mathrm{max}(\hat{C})$ and $S_\mathrm{max}(H_{s,t})$ do not intersect. Then, given that $S_\mathrm{max}(\hat{C})$ is the last cut found by our algorithm, to mark node $t$ there must exist a non-path edge connecting the endpoint $v$ of some edge $e = (u, v) \in S_\mathrm{max}(\hat{C})$ to $t$. But this implies that no path-successor of edge $e$ can be in an $s$-$t$ mincut, which makes $e$ the rightmost edge on its path that belongs to an $s$-$t$ mincut. Therefore, $e$ must also be contained in $S_\mathrm{max}(H_{s,t})$, a contradiction. 

It only remains to show Property \ref{property.I}, which states that the collection $\hat{C}$ is of maximum cardinality $k_{\textbf{max}}$. For this, we make the following claim, whose proof is analogous to the proof of Property \ref{property.III}. Let $\hat{C}_i$ be the collection of $s$-$t$ mincuts maintained by the algorithm at the end of iteration $i$. 

\begin{claim} \label{claim:leftmostCuts}
Consider set $\hat{C}_{i-1}$ and let $X_i$ be the minimum $s$-$t$ cut found by the algorithm at iteration $i$. Then, there is no minimum $s$-$t$ cut $Y$ such that: (i) $Y$ is disjoint from each $X \in \hat{C}_{i-1}$, and (ii) $Y$ contains an edge that is a path predecessor of an edge of $X_i$.
\end{claim}

In other words, as the algorithm makes progress, no minimum $s$-$t$ cut---that is disjoint from the ones found so far by the algorithm---has edges to the left of the minimum $s$-$t$ cut found by the algorithm at the present iteration. Next, we show that this implies the maximality of the size of the solution returned by the algorithm. 

Let $C_\mathrm{max}$ be a maximum-sized collection of $s$-$t$ mincuts in the graph. Without loss of generality, assume that $C_\mathrm{max}$ is in left-right order (otherwise, by Corollary \ref{corollary:solutionSpace} we can always obtain an equivalent collection that is left-right ordered) and that $S_\mathrm{min}(H_{s, t}) \in C_\mathrm{max}$ and $S_\mathrm{max}(H_{s, t}) \in C_\mathrm{max}$. For the sake of contradiction, suppose that the collection $\hat{C}$ returned by our algorithm is of cardinality $|\hat{C}| = \ell < k_\mathrm{max}$. 

\begin{observation} \label{obs.new}
    There exists at least one minimum $s$-$t$ cut $Y \in C_\mathrm{max}$ such that $X_{i} < Y$ and $Y$ contains at least one edge that is a path predecessor of an edge in $X_{i+1}$, with $X_{i}$ and $X_{i+1}$ two consecutive cuts in $\hat{C}$.
\end{observation}
\begin{proof}
    Let $C_\mathrm{max} = \{Y_1, Y_2, \ldots, Y_{k_\mathrm{max}}\}$, where $Y_1 = S_\mathrm{min}(H_{s, t})$ and $Y_{k_\mathrm{max}} = S_\mathrm{max}(H_{s, t})$, and let $\hat{C} = \{X_1, \ldots, X_\ell\}$. We know by Property \ref{property.III} that $X_1 = S_\mathrm{min}(H_{s, t})$. Hence, there is always an $s$-$t$ mincut in $C_\mathrm{max}$ that is a strict successor of a cut in $\hat{C}$, namely $Y_2 > X_1$. For the sake of contradiction, suppose that the observation is false. Then, every cut $Y \in C_\mathrm{max}$ that is a strict successor of a cut $X_i \in \hat{C}$ is also a (not necessarily strict) successor of the cut $X_{i+1} \in \hat{C}$, for $i\in \{1, \ldots, \ell -1\}$. Let this be true for the first $\ell -1$ cuts of $\hat{C}$. Then, the last $k_\mathrm{max} - \ell$ cuts of $C_\mathrm{max}$ must be disjoint from the first $\ell -1$ cuts of $\hat{C}$. The last cut $X_\ell$ of $\hat{C}$ must then be located in or before the gap between the first $\ell$ cuts in $C_\mathrm{max}$ and its remaining $k - \ell$ cuts. But we know by Property \ref{property.IV} that $X_\ell \cap Y_{k_\mathrm{max}} \neq \emptyset$, which gives the necessary contradiction.      
\end{proof}

Observation \ref{obs.new} stands in contrast with Claim \ref{claim:leftmostCuts}, which states that such a cut $Y$ cannot exist. Hence, we obtain a contradiction, and the collection $\hat{C}$ returned by the algorithm must be of maximum cardinality. This completes the proof of Lemma \ref{lemma.algorithmCorrectness}. 

\subsection{Handling the general case} \label{sec.generalCase}

% We now consider \textsc{Max-Disjoint MC} in general graphs. Recall from the previous subsection that, from a graph $G$, one can construct a path graph $H_{s, t}$ such that every minimum $s$-$t$ cut in $G$ is also a minimum $s$-$t$ cut in $H_{s, t}$. Ideally, we would like to use Algorithm \ref{algo.kdisjointMSP} in $H_{s, t}$ to solve \textsc{Max-Disjoint MC} in $G$. But, as we argued, the path graph $H_{s,t}$ may not have the same set of $s$-$t$ mincuts as $G$. Here we show that $H_{s,t}$ can be augmented with edges such that its minimum $s$-$t$ cuts correspond bijectively to those in $G$, which serves to solve the general problem.
We now consider \textsc{Max-Disjoint MC} in general graphs. Recall from the previous subsection that, from a graph $G$, one can construct a path graph $H_{s, t}$ such that every minimum $s$-$t$ cut in $G$ is also a minimum $s$-$t$ cut in $H_{s, t}$. We could propose to use Algorithm \ref{algo.kdisjointMSP} in $H_{s, t}$ to solve \textsc{Max-Disjoint MC} in $G$. But, as we argued previously, the path graph $H_{s,t}$ may not have the same set of $s$-$t$ mincuts as $G$. We can, however, solve this challenge by augmenting $H_{s,t}$ with edges such that its minimum $s$-$t$ cuts correspond bijectively to those in $G$.

% \begin{definition} \label{definition:stPathGraph}
% \normalfont An \textit{augmented $s$-$t$ path graph} of $G$ is the subgraph $H_{s,t}(G)$ induced by the set $V(\mathcal{P}_{s,t}(G))$, with additional %non-path 
% edges between any two vertices $u, v \in V(H_{s,t}(G))$ %such that 
% if $v$ is reachable from $u$ in $G$ by a path whose internal vertices are exclusively in $V(G) \setminus V(H_{s,t}(G))$. 
% \end{definition}

% In view of this definition, the following claim and lemma serve as the correctness and complexity proofs of the proposed algorithm for the general case (see proofs in Appendix \ref{appendix.proofsSection4}). 

\begin{definition} \label{definition:stPathGraph}
\normalfont An \textit{augmented $s$-$t$ path graph} of $G$ is a path graph $H_{s,t}(G)$ of height $\lambda(G)$, with additional non-path edges between any two vertices $u, v \in V(H_{s,t}(G))$ such that $v$ is reachable from $u$ in $G$ by a path whose internal vertices are exclusively in $V(G) \setminus V(H_{s,t}(G))$. 
\end{definition}

In view of this definition, the following claim and lemma serve as the correctness and complexity proofs of the proposed algorithm for the general case.

\begin{claim} \label{claim:pathGraphSameSet}
An augmented $s$-$t$ path graph of $G$ has the same set of $s$-$t$ mincuts as $G$. 
\end{claim}
\begin{proof}
    By Menger's theorem, we know that a minimum $s$-$t$ cut in $G$ must contain exactly one edge from each path in $\mathcal{P}_{s, t}(G)$, where $|\mathcal{P}_{s, t}(G)| = |\lambda(G)|$. W.l.o.g., let $H_{s, t}(G)$ be the augmented $s$-$t$ path graph of $G$ such that each path $p \in \mathcal{P}_{s, t}(G)$ is also present in $H_{s, t}(G)$. 
We now show that a minimum $s$-$t$ cut in $G$ is also present in $H_{s, t}(G)$. The argument in the other direction is similar and is thus omitted. 

Consider an arbitrary minimum $s$-$t$ cut $X$ in $G$. For the sake of contradiction, assume that $X$ is not a minimum $s$-$t$ cut in $H_{s, t}(G)$. Then, after removing every edge of $X$ in $H_{s, t}(G)$, there is still at least one $s$-$t$ path left in the graph. Such a path must contain an edge $(u, v)$ such that $u \leq w$ and $w' \leq v$, where $w$ and $w'$ are the tail and head nodes of two (not necessarily distinct) edges in $X$, respectively. By definition of $H_{s, t}(G)$, there is a path from $u$ to $v$ in $G$ that does not use edges in $\mathcal{P}_{s, t}(G)$. But then removing the edges of $X$ in $G$ still leaves an $s$-$t$ path in the graph. Thus $X$ cannot be an $s$-$t$ cut, and we reach our contradiction.
\end{proof}

% \begin{lemma} \label{lemma:stPathGraphConstruction}
% An augmented $s$-$t$ path graph $H$ of a graph $G$ can be constructed in time $O(F(m, n) + m\lambda(G))$, where $F(m, n)$ is the time required by a max-flow computation.
% \end{lemma}
\begin{lemma} \label{lemma:stPathGraphConstruction}
An augmented $s$-$t$ path graph $H$ of a graph $G$ can be constructed in time $O(F(m, n) + m\lambda(G))$, where $F(m, n)$ is the time required by a max-flow computation.
\end{lemma}
\begin{proof}
    The idea of the algorithm is rather simple. First, we find a maximum cardinality collection of edge-disjoint $s$-$t$ paths in $G$ and take their union to construct a ``skeleton'' graph $H$. Then, we augment the graph by drawing an edge between two path vertices $u, v \in H$ if $v$ is reachable from $u$ in $G$ by using exclusively non-path vertices. By definition, the resulting graph is an augmented $s$-$t$ path graph of $G$.

Now we look into the algorithm's implementation and analyze its running time. It is folklore knowledge that the problem of finding a maximum-sized collection of edge-disjoint $s$-$t$ paths in a graph with $n$ vertices and $m$ edges can be formulated as a maximum flow problem. Hence, the first step of the algorithm can be performed in $F(m, n)$ time. Let $\mathcal{P}_{s, t}(G)$ denote such found collection of $s$-$t$ paths. 

The second step of the algorithm could be computed in $O(mn)$ time by means of an \textit{all-pairs reachability} algorithm. Notice, however, that for a path vertex $v$ all we require for a correct execution of Algorithm \ref{algo.kdisjointMSP} is knowledge of the rightmost vertices it can reach on each of the $\lambda(G)$ paths (Line 6 of Algorithm \ref{algo.kdisjointMSP}). Hence, we do not need to draw every edge between every pair of reachable path vertices, only $\lambda(G)$ edges per vertex suffice. This can be achieved in $O(m \lambda(G))$ time as follows.   

In the original graph, first equip each vertex $u \in V(G)$ with a set of $\lambda(G)$ variables $R(p, u)$, one for each path $p\in \mathcal{P}_{s, t}(G)$. These variables will be used to store the rightmost vertex $v \in p$ that is reachable from $u$. Next, consider a path $p \in \mathcal{P}_{s, t}(G)$ represented as a sequence $[v_1, v_2, \ldots, v_p]$ of internal vertices (i.e., with $s$ and $t$ removed). For each vertex $v \in p$, in descending order, execute the following procedure \texttt{propagate($v$, $p$)}: Find the set $N(v)$ of incoming neighbors of $v$ in $G$ and, for each $w \in N(v)$ if $R(p, w)$ has not been set, let $R(p, w) = v$ and mark $w$ as visited. 
Then, for each node $w \in N(v)$, if $w$ is an unvisited non-path vertex, execute \texttt{propagate($w$, $p$)}; 
otherwise, do nothing. Notice that, since we iterate from the rightmost vertex in $p$, any node $u$ such that $R(u, p) = v_i$ cannot change its value when executing \texttt{propagate($v_j$)} with $j < i$. In other words, each vertex only stores information about the rightmost vertex it can reach in $p$. Complexity-wise, every vertex $v$ in $G$ will be operated upon at most $deg(v)$ times. Hence, starting from an unmarked graph, a call to \texttt{propagate($v$, $p$)} takes $O(m)$ time. Now, we want to execute the above for each path $p \in \mathcal{P}_{s, t}(G)$ (unmarking all vertices before the start of each iteration). This then gives us our claimed complexity of $O(m \lambda(G))$. 
The claim follows from combining the running time of both steps of the algorithm. 
\end{proof}

The following is an immediate consequence of Lemma \ref{lemma:stPathGraphConstruction} and Claim \ref{claim:sweepingAlgoComplexity}. 

\begin{corollary} \label{corollary:disjointGeneral}
There is an algorithm that, given a graph $G$ and two specified vertices $s, t \in V(G)$, in $O(F(m, n) + m\lambda(G))$ time finds a collection of maximum cardinality of pairwise disjoint $s$-$t$ mincuts in $G$. 
\end{corollary}

By replacing $F(m, n)$ in Corollary \ref{corollary:disjointGeneral} with the running time of the current best algorithms of \cite{liu2020faster, kathuria2020potential} for finding a maximum flow, we obtain the desired running time of Theorem \ref{thm:3}.

\section{Hardness of \texorpdfstring{\textsc{Min-$k$-DMC}}{Min-k-DMC}}\label{sec.hardness}
In contrast to the polynomial-time algorithms of the previous sections, here we show that \textsc{$k$-DMC} is NP-hard when considering $d_\mathrm{min}$ as the diversity measure. We called this variant \textsc{Min-$k$-DMC} in Section \ref{sec.introduction}. The hardness proof is split into three parts. In Section \ref{sec.HalfMinVCMB}, we first show that a variant of the \textit{constrained minimum vertex cover} problem on bipartite graphs (\textsc{Min-CVCB}) of Chen and Kanj \cite{chen2003constrained} is NP-hard. Next, in Section \ref{sec.firstReduc} we give a reduction from this problem to \textsc{2-Fixed 3-DMC}, a constrained version of \textsc{$3$-DMC}. Finally, in Section \ref{sec.hardnessProof}, we give a polynomial time reduction from \textsc{2-Fixed 3-DMC} to \textsc{Min-$3$-DMC}, which completes the proof that \textsc{Min-$k$-DMC} is NP-hard.

%\subsection{Hardness of \textsc{Balanced Min-VCMB}} \label{sec.HalfMinVCMB}
\subsection{A first reduction} \label{sec.HalfMinVCMB}
Let us first introduce the \textit{constrained minimum vertex cover} problem on bipartite graphs. 

\begin{extthm}[Min-CVCB]
    Given a bipartite graph $G = (V, E)$ with bipartition $V = A \cup B$ and two positive integers $k$ and $\ell$, determine whether there is a minimum vertex cover $V'$ of $G$ such that $|V' \cap A| \leq k$ and $|V' \cap B| \leq \ell$.
\end{extthm}

This problem was proven to be NP-hard by Chen and Kanj \cite[Thm. 3.1]{chen2003constrained} via a reduction from the NP-complete problem \textsc{Clique} \cite{garey1979computers}. In the \textsc{Min-CVCB} instance $\langle G, k, l\rangle$ they construct, the bipartite graph $G$ has a perfect matching. Thus, the following can be stated as a corollary of their reduction. We refer to a bipartite graph that has a perfect matching as a \textit{matched bipartite graph}. 

\begin{corollary} \label{corollary.1}
    \textsc{Min-CVCB} is NP-hard even in matched bipartite graphs. 
\end{corollary}

Furthermore, we can prove the hardness of the following variant of the problem. 

\begin{extthm}[\textsc{Balanced Min-VCMB}]
    Given a matched bipartite graph $G = (V, E)$ with bipartition $V = A \cup B$, determine whether there is a minimum vertex cover $V'$ of $G$ such that $|V' \cap A| = |V' \cap B|$.
\end{extthm}

\begin{lemma}
    \textsc{Balanced Min-VCMB} is NP-hard.
\end{lemma}
\begin{proof}
    This problem is a variant of \textsc{Min-CVCB} with two added restrictions: (i) the input graph $G$ has a perfect matching, and (ii) $k = \ell = |V|/4$. Corollary \ref{corollary.1} already states the hardness of the variant with only (i) as an added restriction. Thus, we need only argue that the problem remains hard when also considering restriction (ii). 

    Let $\langle G = (V, E), k, l \rangle$ be an instance of \textsc{Min-CVCB} with bipartition $V = A \cup B$, where $G$ has a perfect matching. We will construct an instance $\langle G' = (V', E') \rangle$ with bipartition $V' = A'\cup B'$ of \textsc{Balanced Min-VCMB} and show that $G$ has a feasible solution iff $G'$ has a balanced minimum vertex cover $C$; i.e., one such that $|C \cap A'| = |C \cap B'|$. First, observe that because $G$ is a matched bipartite graph, we have $|A| = |B| = |V|/2$. Moreover, the size of a minimum vertex cover in $G$ is exactly $k + l = |V|/2$. We distinguish three %\footnote{Because $k + l = |v|/2$, the third case $k > |V| / 4$ can be seen as case (b) by replacing $k$ with $\ell$.} 
    cases: (a) $k > |V| / 4$, (b) $k < |V| / 4$, and (c) $k = |V| / 4$. Case (c) is trivial, and since $k + \ell = |V|/2$, case (a) can be seen as case (b) by replacing $k$ with $\ell$. So, we will focus on proving case (b). 

    Starting from graph $G$, we construct $G'$ by adding $t = |V|/2 - 2k$ dummy pairs of vertices $\{x_1, y_1\}, \ldots, \{x_t, y_t\}$ to $G$ and connecting them as follows. For each pair $\{x_i, y_i\}$ we create the edge $(x_i, y_i)$ and add an edge from $x_i$ to every vertex in $B$. Let $X = \bigcup_{1 \leq i \leq t} x_i $ and $Y = \bigcup_{1 \leq i \leq t} y_i$. Then, the bipartition of $G'$ is given by $A' = A \cup X$ and $B' = B \cup Y$. Observe that $G'$ is matched bipartite since the subgraph induced by $X$ and $Y$ has a perfect matching. Clearly, the size of a minimum vertex cover in $G'$ is $|V| + t$. We further claim that any minimum vertex cover $C$ in $G'$ such that $|C \cap B'| < |V|/2$ satisfies that $X \subset C$. It is easy to see this by contradiction: since $|C \cap B'| < |V|/2$, there is at least one edge that has no endpoint in $C$, namely an edge connecting a vertex in $B' \setminus C$ to a vertex in $X \setminus C$. 
    %since $|C\cap B| \leq |C\cap B'| < |V|/2 = |B|$, there is a vertex $v\in B$ that is not in the cover $C$. Since $v$ is connected to every vertex in $X$, this implies that $X\subset C$.

    Now, we are ready to prove that $G$ has a minimum vertex cover $D$ with $|A\cap D| = k$ iff $G'$ has a minimum vertex cover $C$ such that $|C \cap A'| = |C \cap B'| = \frac{|V|/2 + t}{2}$. 

    Let $D$ be a minimum vertex cover in $G$ such that $|D \cap A| = k$. We claim that $C = D \cup X$ is a solution to \textsc{Balanced Min-VCMB} in $G'$. This follows because, by construction of $X$, the set $C$ contains an endpoint of every edge in $G'$. Moreover, $|C \cap A'| = k + t = \frac{|V|/2 + t}{2}$, which follows from $k = \frac{|V|/2 - t}{2}$. 

    Conversely, given a minimum vertex cover $C$ in $G'$ such that $|C \cap A'| = |C \cap B'| = \frac{|V|/2 + t}{2}$, we claim that $D = C \setminus X$ is a minimum vertex cover in $G$ with $k$ $A$-vertices. The cardinality part of the claim follows from the earlier observed fact that $X \subset C$, hence $|D \cap A| = \frac{|V|/2 - t}{2} = k$. On the other hand, because both $G$ and $G'[X \cup Y]$ (the subgraph of $G'$ induced by $X\cup Y$) have perfect matchings, the intersection of $C$ with (the vertex set of) each subgraph is a minimum vertex cover for the subgraph. Hence, $D$ is a minimum vertex cover for $G$, proving the claim.     

    Since the construction of the instance $\langle G' = (V', E') \rangle$ of \textsc{Balanced Min-VCMB} can clearly be performed in polynomial time, the theorem is proved.
\end{proof}

\subsection{A second reduction} \label{sec.firstReduc}
In this section, we introduce a variant of \textsc{3-DMC}, and establish its NP-hardness. 

\begin{extthm}[\textsc{2-Fixed 3-DMC}]
    Given are two minimum $s$-$t$ cuts $X$ and $Y$ in a graph $G = (V, E)$, and an integer $\ell > 0$. Answer whether there is a third minimum $s$-$t$ cut $S$ in $G$ such that $\max(|X \cap S|, |Y \cap S|) \leq \ell$. 
\end{extthm}

%A bipartite graph $G = (V, E)$ whose bipartition sets $A$ and $B$ have equal cardinality $|V|/2$ is called a \textit{balanced bipartite graph}. 

We concentrate on a special type of \textsc{2-Fixed 3-DMC} where the graph $G = (V, E)$ satisfies the following restrictions: (i) any maximum cardinality set of (internally) vertex disjoint $s$-$t$ paths covers all the vertices in $V$, and (ii) any $s$-$t$ path has length 3. With a slight abuse of terminology, we refer to such a graph as a \textit{bipartite $s$-$t$ graph}. Notice that $G \setminus \{s, t\}$ is a matched bipartite graph. Showing the NP-hardness of \textsc{2-Fixed 3-DMC} in bipartite $s$-$t$ graphs clearly implies hardness in the general case. %Furthermore, we can exploit the special structure of this graph in subsequent reductions.  

\begin{lemma} \label{theorem.1}
    \textsc{2-Fixed 3-DMS} is NP-hard, even in bipartite $s$-$t$ graphs. 
\end{lemma}
\begin{proof}
%Clearly, \textsc{2-Fixed 3-DMS} is in NP. 
We show that \textsc{2-Fixed 3-DMS} is NP-hard by reducing from \textsc{Balanced Min-VCMB}. %See Appendix \ref{appendix.HalfMinVCMB} for the corresponding hardness proof. 
%Consider first the following problem: 
% \begin{extthm}[\textsc{Balanced Biclique}]
%     Given a bipartite graph $G = (V, E)$ with $V = A \cup B$ and a non-negative integer $k \leq |V|$, are there two disjoint subsets $V_1, V_2 \subseteq V$ satisfying $|V_1| = |V_2| = k$ such that $u \in V_1$, $v \in V_2$ implies that $(u, v) \in E$?
% \end{extthm}
%\textsc{Half Balanced Biclique} is a variant of the \textsc{Balanced Biclique} problem \cite[GT24]{garey1979computers} with two added restrictions: (i) the two parts of the bipartition of $G$ have the same size, and (ii) the sought biclique in $G$ is of size exactly $|V|/4$. 
%There is a simple reduction from \textsc{Balanced Biclique} \cite[GT24]{garey1979computers} to \textsc{Half Balanced Biclique} that shows the latter problem to be NP-hard (see the appendix for details). %The proof is a modification of the reduction from \textsc{Clique} to \textsc{Fractional Clique} \cite[GT19]{garey1979computers}. 
%We show that \textsc{2-Fixed 3-DMS} is NP-hard by reducing from \textsc{Half Balanced Biclique}. 
%Given an instance XX of \textsc{Half Balanced Biclique} (we may assume that there are no same-level edges), 
Let $\langle G = (V, E) \rangle$ with $V = A \cup B$ be an instance of \textsc{Balanced Min-VCMB}, where $m := |V|/2$. Since $G$ has a perfect matching, we have $|A| = |B| = m$. We construct an instance $\langle H, X, Y, \ell \rangle$ of \textsc{2-Fixed 3-DMC} %, where $H$ is a bipartite $s$-$t$ graph, $X$ and $Y$ are minimum $s$-$t$ cuts, and $\ell$ is a nonnegative integer, 
as follows. 
% ---with $\ell = m/2$---
% and show that there is a positive solution to \textsc{HBB} iff  there is a minimum $s$-$t$ separator $S$ in $H$ such that $|X \cap S| = |Y \cap S| = m/2$. 
The directed graph $H$ simply consists of a copy of $G$, with edges directed from $A$ to $B$, plus two additional vertices $s$ and $t$, such that $s$ has an outgoing arc to every vertex in $A$ and $t$ has an incoming arc from every vertex in $B$. %This construction can clearly be performed in polynomial time. 
%See Figure \ref{fig:reductionGraph} for an illustration. 
For the fixed cuts $X$ and $Y$, we let $X = \{(s, u) \; | \; u \in A\}$ and $Y = \{(v, t) \; | \; v \in B\}$. Clearly, these are $s$-$t$ cuts. The fact that they are of minimum size stems from $G$ having a perfect matching of size $m$ since this implies that $H$ has at least $m$ (internally) edge-disjoint $s$-$t$ paths. 
Finally, we set $\ell = m/2$. We claim that a minimum vertex cover $V'$ of size $m$ exists in $G$ such that $|V' \cap A| = |V' \cap B| = m/2$ iff there is a minimum $s$-$t$ cut $S$ in $H$ such that $|X \cap S| = |Y \cap S| = m/2$. On the one hand, this follows from observing that, by definition of vertex cover, the set of edges that connect $s$ and $t$ to the vertices of a minimum vertex cover in $G$ is always a minimum $s$-$t$ cut in $H$. %and vice versa. %(since $G$ has a perfect matching, every minimum vertex cover has size $m$ and, by construction, there are exactly $m$ vertex-disjoint $s$-$t$ paths in $H$, so any minimum $s$-$t$ separator in $H$ has size $m$). 
On the other hand, a minimum $s$-$t$ cut $S$ in $G$ satisfying $|X \cap S| = |Y \cap S| = m/2$ cannot contain an edge $(u, v)$ such that $u \in A$ and $v \in B$. Otherwise, the intersection size would be less than $m/2$. Then, %by definition of cut, 
it follows that the vertex set $\{u \; | \; (s, u) \in S\} \cup \{v \; | \; (v, t) \in S\}$ is a minimum vertex cover in $G$. 
%Thus, equating $V'$ to $S$ suffices for the two directions of the claim. 
\end{proof}

\subsection{The final reduction} \label{sec.hardnessProof}
We now present the NP-hardness proof for the decision version of \textsc{Min-$k$-DMC}. For simplicity, we consider \textsc{Min-$k$-DMC} reformulated as a minimization problem by means of the relationship $\max_{S \in U_k} d_\mathrm{min}(S) = \min_{S \in U_k} \hat{d}_\mathrm{min}(S)$, where

\begin{equation*}
    \hat{d}_{\mathrm{min}}(X_1, \ldots, X_k) = \max_{1\leq i \leq j \leq k} |X_i \cap X_j|.
\end{equation*}

\begin{extthm}[Min-$k$-DMC (decision version)]
    Given are a directed graph $G = (V, E)$, vertices $s,t \in V$, and two integers $k, \ell > 0$. Let $\Gamma_G(s, t)$ be the set of minimum $s$-$t$ cuts in $G$, and let $U_k$ be the set of $k$-element multisets of $\Gamma_G(s, t)$. Answer whether $G$ has a collection $C \in U_k$ of minimum $s$-$t$ cuts such that $\hat{d}_\mathrm{min}(C) \leq \ell$. 
\end{extthm}

%More precisely, we show that the problem is already NP-hard when $k = 3$. 

\begin{theorem}
    The decision version of \textsc{Min-$k$-DMC} is NP-hard, even for $k = 3$. 
\end{theorem}
\begin{proof}
    We give a polynomial time reduction from \textsc{2-Fixed 3-DMC} in bipartite $s$-$t$ graphs, which is proven NP-hard in Theorem \ref{theorem.1}. Let $\langle G, X, Y, \ell \rangle$ be an instance of \textsc{2-Fixed 3-DMC} with bipartition\footnote{Whenever we refer to the bipartition of a bipartite $s$-$t$ graph $G$, we mean the bipartition of its induced bipartite graph $G \setminus \{s, t\}$.} $A \cup B$ %Similar to the proof of Theorem \ref{theorem.1}, 
    and let the vertices of $A$ (resp. $B$) be labeled $a_1,\ldots,a_m$ (resp. $b_1,\ldots,b_m$), 
    %we assume that the vertices of $A$ (resp. $B$) are labeled $a_1$ (resp. $b_1$) through $a_m$ (resp. $b_m$), 
    where $m := |A \cup B|/2$. Furthermore, let $E_A = \{(s, v) \; | \; v \in A\}$ and $E_B = \{(v, t) \; | \; v \in B\}$ we fix $X = E_A$ and $Y = E_B$ and let $\ell = m/2$. We can do this since we proved in Theorem \ref{theorem.1} that \textsc{2-Fixed 3-DMC} is NP-hard even in the special case where $X = E_A$ and $Y = E_B$. For the sake of simplicity, w.l.o.g. assume that there is an edge $(a_i, b_i) \in G$ for each $i \in \{1, \ldots, m\}$. %and assume that $m$ is even. 
    Observe that the minimum $s$-$t$ cut size in $G$ is $m$. 

    Given $\langle G, X, Y, \ell \rangle$, we construct an instance $\langle H, s, t, \ell' \rangle$ of \textsc{Min-$3$-DMC} as follows. The graph $H$ consists of a copy of $G$ plus some additional vertices and edges. The first of these additions are two vertices $u$ and $w$, such that $s$-$u$-$w$-$t$ is a path in $H$. Moreover, we draw the edges $(a_i, u)$ and $(w, b_i)$ for each $i \in \{1, \ldots, m\}$. The second addition is a balanced complete bipartite graph $H_{b} = (A_b \cup B_b, E')$ such that $|A_b| = |B_b| = m/2$. Besides the edges inside $H_b$, we draw the edges $\{(s, a') \; | \; a' \in A_b\}$ and $\{(b', t) \; | \; b' \in B_b\}$. Finally, we set $\ell' = \ell = m/2$. Clearly, this construction can be completed in polynomial time. See Figure \ref{fig:reductionGraph2} for an illustration. Notice that the minimum $s$-$t$ cut size in $H$ is $\frac{3m}{2} + 1$. Furthermore, we claim the following. 

    \begin{claim} \label{claim:3}
        Any minimum $s$-$t$ cut in $H$ that contains the edge $(s, u)$ (resp. $(w, t)$) must also contain every edge in $E_A$ (resp. $E_B$). 
    \end{claim}
    \begin{proof}
        First, we show that the edge $(u, w)$ is crossed by some $s$-$t$ path $p^*$ of any maximum cardinality set of edge-disjoint $s$-$t$ paths in $H$. Let $\mathcal{P}_{s, t}$ denote an arbitrary such set. %Let $\mathcal{P}_{s, t}$ be a maximum cardinality set of edge-disjoint $s$-$t$ paths in $H$. 
        We know that the cardinality of $\mathcal{P}_{s, t}$ is the size of a minimum $s$-$t$ cut in $H$; namely $3 \frac{m}{2}+1$. Then, because there are exactly $3 \frac{m}{2}+1$ outgoing edges from $s$, there is a path $p^* \in \mathcal{P}_{s, t}$ that visits vertex $u$. Moreover, $u$ has only one outgoing edge to $w$, which implies that $(u, w)$ is contained in $p^*$. 

        Back to the main claim, let $C$ be a minimum $s$-$t$ cut in $H$ containing the edge $(s, u)$. For the sake of contradiction, suppose there is an edge $(s, a)$ for some $a \in A$ such that $(s, a) \not\in C$. Then, there is an $s$-$t$ path that crosses no edge in $C$; namely, the path $s$-$a$-$u$-$w$-$b$-$t$, with $b \in B$ a vertex visited by $p^*$. Therefore, the edge set $C$ is not an $s$-$t$ cut, which gives the necessary contradiction.
    \end{proof}
    
    Now we prove that there is a feasible solution to \textsc{2-Fixed 3-DMS} in $G$ iff there are three minimum $s$-$t$ cuts $(P, Q, R)$ in $H$ such that $|P \cap Q| = |P \cap R| = |Q \cap R| = m/2$.

    \begin{figure}
\centering
\begin{tikzpicture}
\tikzset{edge/.style = {->,> = latex'}}
% Nodes
\node[label={$s$},draw=black,circle,fill] (s) at (0, 0) {};
\node[draw=black,circle] (x1) at (1.5, 1.5) {};
\node[draw=black,circle] (x2) at (1.5, 0.5) {};
\node[draw=black,circle] (x3) at (1.5, -0.5) {};
\node[draw=black,circle] (x4) at (1.5, -1.5) {};

\node[draw=black,circle] (y1) at (3, 1.5) {};
\node[draw=black,circle] (y2) at (3, 0.5) {};
\node[draw=black,circle] (y3) at (3, -0.5) {};
\node[draw=black,circle] (y4) at (3, -1.5) {};

\node[label={$t$},draw=black,circle,fill] (t) at (4.5, 0) {};

%\node at (2.25, 3.25) {$G$};
\node at (2.25, -3.0) {};

% Edges
\draw[edge, line width=2pt, palecerulean] (s) -- (x1);
\draw[edge, line width=2pt, bicolor={palecerulean}{color2}] (s) -- (x2);
\draw[edge, line width=2pt, bicolor={palecerulean}{color2}] (s) -- (x3);
\draw[edge, line width=2pt, palecerulean] (s) -- (x4);

\draw[edge] (x1) -- (y1);
\draw[edge] (x2) -- (y1);
\draw[edge] (x2) -- (y2);
\draw[edge] (x3) -- (y3);
\draw[edge] (x2) -- (y3);

\draw[edge] (x4) -- (y4);
\draw[edge] (x3) -- (y2);

\draw[edge, line width=2pt, bicolor={bubblegum}{color2}] (y1) -- (t);
\draw[edge, line width=2pt, bubblegum] (y2) -- (t);
\draw[edge, line width=2pt, bubblegum] (y3) -- (t);
\draw[edge, line width=2pt, bicolor={bubblegum}{color2}] (y4) -- (t);
\end{tikzpicture}
\hspace{6em}
\begin{tikzpicture}
\tikzset{edge/.style = {->,> = latex'}}
% Nodes
\node[label={$s$},draw=black,circle,fill] (s) at (0, 0) {};
\node[draw=black,circle] (x1) at (1.5, 1.5) {};
\node[draw=black,circle] (x2) at (1.5, 0.5) {};
\node[draw=black,circle] (x3) at (1.5, -0.5) {};
\node[draw=black,circle] (x4) at (1.5, -1.5) {};

\node[draw=black,circle] (y1) at (4,1.5) {};
\node[draw=black,circle] (y2) at (4,0.5) {};
\node[draw=black,circle] (y3) at (4,-0.5) {};
\node[draw=black,circle] (y4) at (4,-1.5) {};

\node[draw=black,circle] (x01) at (1.5, 3.5) {};
\node[draw=black,circle] (x02) at (1.5, 2.5) {};
\node[draw=black,circle] (y01) at (4,3.5) {};
\node[draw=black,circle] (y02) at (4,2.5) {};

\node[label={$t$},draw=black,circle,fill] (t) at (5.5,0) {};

% In gray
% \node[label=below:{$u$},draw=black,circle,fill=gray!66] (u) at (1.5, -2.5) {};
% \node[label=below:{$v$},draw=black,circle,fill=gray] (v) at (2.25, -2.5) {};
% \node[label=below:{$w$},draw=black,circle,fill=gray!66] (w) at (3, -2.5) {};

\node[label=below:{$u$},draw=black,circle] (u) at (2.25,-2.5) {};
\node[label=below:{$w$},draw=black,circle] (w) at (3.25,-2.5) {};

%\node at (2.25, 3.25) {$H$};

% Edges
\draw[edge, line width=2pt, palecerulean] (s) -- (x1);
\draw[edge, line width=2pt, bicolor={palecerulean}{color2}] (s) -- (x2);
\draw[edge, line width=2pt, bicolor={palecerulean}{color2}] (s) -- (x3);
\draw[edge, line width=2pt, palecerulean] (s) -- (x4);
\draw[edge, line width=2pt, palecerulean] (s) to [bend right=45] (u);
\draw[edge, line width=2pt, bicolor={bubblegum}{palecerulean}] (s) -- (x01);
\draw[edge, line width=2pt, bicolor={bubblegum}{palecerulean}] (s) -- (x02);

\draw[edge] (x1) -- (y1);
\draw[edge] (x2) -- (y1);
\draw[edge] (x2) -- (y2);
\draw[edge] (x3) -- (y3);
\draw[edge] (x2) -- (y3);

\draw[edge] (x4) -- (y4);
\draw[edge] (x3) -- (y2);

\draw[edge, line width=2pt, bicolor={bubblegum}{color2}] (y1) -- (t);
\draw[edge, line width=2pt, bubblegum] (y2) -- (t);
\draw[edge, line width=2pt, bubblegum] (y3) -- (t);
\draw[edge, line width=2pt, bicolor={bubblegum}{color2}] (y4) -- (t);
\draw[edge, line width=2pt, bubblegum] (w) to [bend right=45] (t);
\draw[edge, line width=2pt, color2] (y01) -- (t);
\draw[edge, line width=2pt, color2] (y02) -- (t);

\draw[edge, line width=2pt, color2] (u) -- (w);

\draw[edge] (x01) -- (y01);
\draw[edge] (x01) -- (y02);
\draw[edge] (x02) -- (y01);
\draw[edge] (x02) -- (y02);

\draw[edge] (x1) -- (u);
\draw[edge] (x2) -- (u);
\draw[edge] (x3) -- (u);
\draw[edge] (x4) -- (u);

\draw[edge] (w) -- (y1);
\draw[edge] (w) -- (y2);
\draw[edge] (w) -- (y3);
\draw[edge] (w) -- (y4);
\end{tikzpicture}
    \caption{An instance $(G, X, Y, \ell)$ of \textsc{2-Fixed 3-DMC} (left) and the constructed instance $(H, s, t, \ell')$ of \textsc{Min-$3$-DMC} (right). A solution for \textsc{2-Fixed 3-DMC} in $G$ with $\ell = m/2$ can be mapped to a solution for \textsc{Min-$3$-DMC} in $H$ with $\ell' = m/2$ and vice versa.}
    \label{fig:reductionGraph2}
\end{figure}
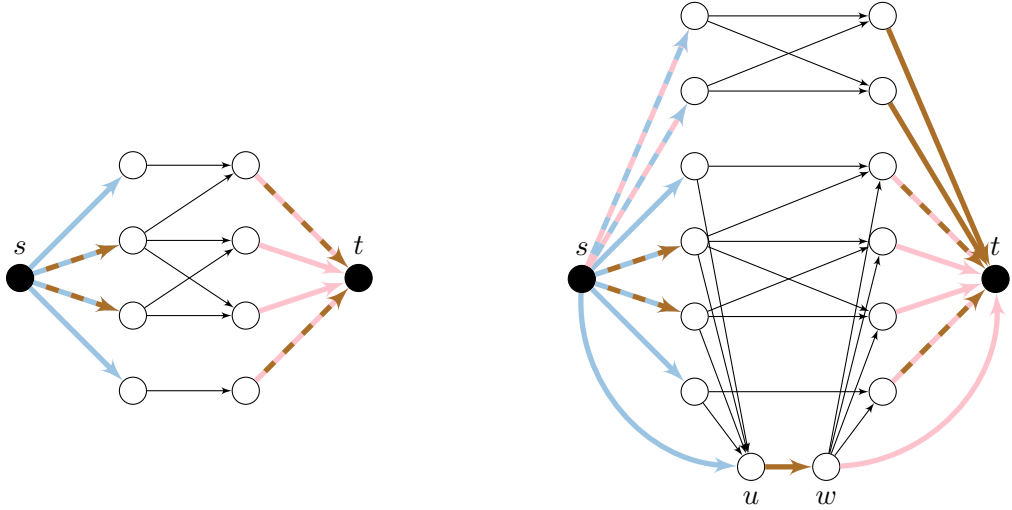

    Suppose that there is a feasible solution to \textsc{2-Fixed 3-DMC} in $G$. That is, there exists a minimum $s$-$t$ cut $S$ in $G$ such that $|X \cap S| = |Y \cap S| = m/2$. Let $P = E^b_A \cup X \cup (s, u)$, $Q = E^b_A \cup Y \cup (w, t)$, and $R = E^b_B \cup S \cup (u, w)$, where $E^b_A = \{(s, a) \; | \; a \in A_b\}$ and $E^b_B = \{(b, t) \; | \; b \in B_b\}$. These cuts are indicated in blue, pink, and brown in the example of Figure \ref{fig:reductionGraph2}, respectively. It is easy to check that these are indeed $s$-$t$ cuts in $H$\footnote{Removing $E^b_A$ or $E^b_B$ prevent $s$ from reaching $t$ via $H_b$. Similarly, removing $(u, w)$ prevents $s$ from reaching $t$ through a path containing $u$ or $w$. Moreover, because there are outgoing (resp. incoming) edges from (resp. to) every vertex in $A$ (resp. $B$) to $u$ (resp. from $w$), any minimum $s$-$t$ cut containing the edge $(s, u)$ (resp. $(w, t)$) must also have $E_A$ (resp. $E_B$) as a subset.} of minimum size $\frac{3m}{2} + 1$. Clearly, the triple $(P, Q, R)$ is a solution to \textsc{Min-$3$-DMC} in $H$ since $|P \cap Q| = |E^b_A| = m/2$, $|P \cap R| = |E_A \cap S| = m/2$, and $|Q \cap R| = |E_B \cap S| = m/2$. 

    Conversely, suppose that there is a solution for \textsc{Min-$3$-DMC} in $H$. That is, there is a triple $(P, Q, R)$ of minimum $s$-$t$ cuts in $H$ such that $|P \cap Q| = |P \cap R| = |Q \cap R| = m/2$. Observe that, by construction, exactly one of the edge sets $E^b_A$ and $E^b_B$ is a subset of any minimum $s$-$t$ cut in $H$. Then, there are two cases (modulo symmetry) for the triple $(P, Q, R)$: (i) $E^b_A \subset P$, $E^b_A \subset Q$, and $E^b_B \subset R$, and (ii) $E^b_A \subset P$, $E^b_A \subset Q$, and $E^b_A \subset R$. Next, we analyze each case and argue that we can always construct a feasible solution $S$ to \textsc{2-Fixed 3-DMC}.

    %We show that case (ii) cannot occur. 
    First, we consider case (ii). Notice that if the three cuts $(P, Q, R)$ all intersect at $E^b_A$, then they must not share any more edges in $G \setminus H_b$ since there would then be a pairwise intersection with a value greater than $m/2$. In particular, this means that the edges $(s, u)$, $(u, w)$ and $(w, t)$ must each belong to exactly one of $P$, $Q$, and $R$. Without loss of generality, assume that $(s, u) \in P$, $(w, t) \in Q$, and $(u, w) \in R$. Then, by Claim \ref{claim:3}, we have that $E_A \subset P$ and $E_B \subset Q$. Now, there are two sub-cases for $R$ to consider. On the one hand, if $R$ contains exactly $m$ edges from $\{(a, b) \; | \; a \in A, b \in B\}$, we have the trivial case of there being exactly $m$ vertex-disjoint $s$-$t$ paths in $G$. Then, we can construct a solution $S$ to \textsc{2-Fixed 3-DMC} by letting $S$ contain exactly one edge from each path, with $m/2$ coming from $E_A$ and $m/2$ from $E_B$. On the other hand, if $R$ contains fewer than $m$ edges from the set $\{(a, b) \; | \; a \in A, b \in B\}$, then $R$ intersects one of $P$ or $Q$ in more than $m/2$ edges, meaning that such a sub-case is not allowed. 
    
    Next, we turn to case (i). Since $E^b_A \subset P$ and $E^b_A \subset Q$, we have that $|P \cap Q| \geq m/2$. Then, to satisfy equality, it must be that $P \cap Q = E^b_A$. This also implies that one of $P$ or $Q$ must contain one of $(s, u)$ or $(w, t)$. %Then, by Claim \ref{claim:4}, one of $P$ or $Q$ is the leftmost or rightmost cut in $H \setminus H_b$. 
    Moreover, $(u, w) \in R$; otherwise, because of Claim \ref{claim:3}, $|P \cap R| \geq m/2$ or $|P \cap R|  = 0$. 
    Now, suppose that $(s, u) \in P$ (the argument for $(w, t) \in P$ is analogous and is thus omitted). Then, by Claim \ref{claim:3}, $P = E_A^b \cup E_A \cup \{(s, u)\}$. We also know that $|P \cap R| = m/2$, and that $P \cap R \subset E_A$. From this, we can already construct a solution $S$ to \textsc{2-Fixed 3-DMC} as follows. Let $D$ be an arbitrary maximum cardinality set of edge-disjoint $s$-$t$ paths in $H$, and let  $\hat{D} \subset D$ be the subset of paths that do not cross an edge in $P \cap R$. Notice that every path in $D$ must each include an edge from $E_B$ since there are exactly $3 \frac{m}{2} +1$ incoming edges to $t$. Now, let $\hat{E}_D \subset E_B$ be the set of edges in $E_B$ that are crossed by paths in $\hat{D}$. We claim that $S = (P \cap R) \cup \hat{E}_D$ is a solution to \textsc{2-Fixed 3-DMC}. To prove this, let $\hat{D}' \subset D$ be the subset of paths that cross an edge in $R$. It is easy to see that $S$ has size $m/2$. The fact that $S$ is an $s$-$t$ cut follows from $R$ being an $s$-$t$ cut in $H$, as there is no edge $(x, y)$ with $x \in A$, $y \in B$, such that $x \in p$, and $y \in q$ with $p \in \hat{D}$ and $q \in \hat{D}'$. Moreover, $|S \cap E_A| = |S \cap E_B| = m / 2$. Hence, $S$ is a feasible solution to \textsc{2-Fixed 3-DMC} with objective value $m/2$. 
\end{proof}

\section{Concluding remarks}\label{sec.conclusion}
We showed that the $k$-\textsc{Diverse Minimum s-t Cuts} problem can be solved efficiently when considering two natural measures for the diversity of a set of solutions. There exist, however, other sensible measures of diversity. One that often arises in literature is the minimum pairwise Hamming distance of a collection of sets. We showed that $k$-\textsc{DMC} is NP-hard when using this \textit{bottleneck} measure as the maximization objective. 
% the \textit{bottleneck} objective. In our context, it consists of maximizing the minimum pairwise Hamming distance of a collection of $s$-$t$ mincuts. %The complexity of $k$-\textsc{DMC} when considering the \textit{bottleneck} objective is still open. The challenge of extending our approach to this measure is that it is not immediately clear how to apply our ordering results to this variant of $k$-\textsc{DMC}.  
% We showed that $k$-\textsc{DMC} is NP-hard when considering the \textit{bottleneck} objective. 
For the special case of finding pairwise-disjoint collections of $s$-$t$ mincuts, we showed that faster algorithms exist when compared to solving $k$-\textsc{DMC} for the pairwise-sum and coverage diversity measures. It is thus natural to ask whether there are faster algorithms for \textsc{Sum}-$k$-\textsc{DMC} and \textsc{Cov}-$k$-\textsc{DMC} (or other variants of $k$-\textsc{DMC}) that do not require the sophisticated framework of submodular function minimization. In this work, we relied on the algebraic structure of the problem to obtain a polynomial time algorithm. We believe it is an interesting research direction to assess whether the notion of diversity in other combinatorial problems leads to similar structures, which could then be exploited for developing efficient algorithms. 

\section*{Acknowledgement}
We thank Martin Frohn for bringing the theory of lattices to our attention, and for fruitful discussions on different stages of this work. 

This research was supported by the European Union’s Horizon 2020 research and innovation programme under the Marie Skłodowska-Curie grant agreement no. 945045, and by the NWO Gravitation project NETWORKS under grant no. 024.002.003.

\bibliography{biblio}

\appendix
\section{Proofs of Section \ref{sec:sfm}} \label{appendix.proofsSection3}

\subsection{Proof of Proposition \ref{proposition:1}} \label{appendix.two}
Before proving the proposition, we require the following claim. 

\begin{claim} \label{claim:leftmostRightmost}
For any $X, Y \in \Gamma_G(s, t)$, we have $S_{\mathrm{min}}(X \cup Y), S_{\mathrm{max}}(X \cup Y) \in \Gamma_G(s, t)$ and $|S_{\mathrm{min}}(X \cup Y) \cap S_{\mathrm{max}}(X \cup Y)| = |X \cap Y|$.
\end{claim}
\begin{proof}
Without loss of generality, let $\mathcal{P}_{s,t}(G)$ be any maximum-sized set of edge-disjoint paths from $s$ to $t$. Recall that, by Menger's theorem, any minimum $s$-$t$ cut in $G$ contains exactly one edge from each path in $\mathcal{P}_{s,t}(G)$.
Thus, for a path $p \in \mathcal{P}_{s,t}(G)$, let $e, f \in p$ be the edges that intersect with cuts $X$ and $Y$, respectively. Then the set $S_\mathrm{min}(X \cup Y)$ can be seen as the subset of $X \cup Y$ where $e \in S_\mathrm{min}(X \cup Y)$ if $e \leq f$, for each path $p \in \mathcal{P}_{s,t}(G)$. Analogous for $S_\mathrm{max}(X \cup Y)$. 

We want to prove that $S_\mathrm{min}(X \cup Y)$ (resp. $S_\mathrm{max}(X \cup Y)$) is an $s$-$t$ cut\footnote{Notice that the size of $S_\mathrm{min}(X \cup Y)$ (resp. $S_\mathrm{max}(X \cup Y)$) is already minimum, as it contains exactly one edge from each path $p \in \mathcal{P}_{s,t}(G)$, which has cardinality $\lambda$.}, and that $|S_{\mathrm{min}}(X \cup Y) \cap S_{\mathrm{max}}(X \cup Y)| = |X \cap Y|$. For the latter, simply observe that whenever $X$ and $Y$ intersect at an edge $e$, by Menger's theorem, the path $p \in \mathcal{P}_{s,t}(G)$ that contains $e$ contains no other edge $f$ from $X \cup Y$. Thus, by %the condition $e \leq e$ (resp. $e \geq e$) 
definition, the edge $e$ will be contained by both $S_\mathrm{min}(X \cup Y)$ and $S_\mathrm{max}(X \cup Y)$. On the other hand, if $S_\mathrm{min}(X \cup Y)$ and $S_\mathrm{max}(X \cup Y)$ intersect at an edge $e'$; by definition, the path from $\mathcal{P}_{s,t}(G)$ containing $e'$ cannot include another edge from $X \cup Y$, since either $S_\mathrm{min}(X \cup Y)$ or $S_\mathrm{max}(X \cup Y)$ would contain it, which we know is not the case. Thus, $e' \in X \cap Y$, and the second part of the claim is proven. %Since these arguments hold for every path $p \in \mathcal{P}_{s,t}(G)$, the second claim follows. 

Now we show that $S_\mathrm{min}(X \cup Y)$ and $S_\mathrm{max}(X \cup Y)$ are $s$-$t$ cuts. We only prove this for $S_\mathrm{min}(X \cup Y)$ since the proof for $S_\mathrm{max}(X \cup Y)$ is analogous. For the sake of contradiction, suppose that $S_\mathrm{min}(X \cup Y)$ is not an $s$-$t$ cut. Then, there exists an $s$-$t$ path $\pi = (s, \ldots, t)$ in $G$ that does not contain en edge from $S_\mathrm{min}(X \cup Y)$. This means that $\pi$ has a subpath $\pi^* = (v_1, \ldots, v_2)$ satisfying $v_1 \leq_p w$ and $w' \leq_q v_2$, where $w$ and $w'$ are, respectively, the head and tail nodes of two (not necessarily distinct) edges $e, f \in S_{\mathrm{min}}(X \cup Y)$, and $p, q \in \mathcal{P}_{s,t}(G)$. 
In other words, there exists a path $\pi^*$ starting at a node $v_1$ which appears before an edge $e \in S_\mathrm{min}$ in a path $p \in \mathcal{P}_{s,t}(G)$, and ending at a node $v_2$ that appears after an edge $f \in S_\mathrm{min}$ in a path $q \in \mathcal{P}_{s,t}(G)$.   
It follows that edge $f$ in path $q$ can never be in an $s$-$t$ cut together with an edge in path $p$ that is to the right of (and including) edge $e$ (unless an edge from $\pi$ is also cut, but then the cut is not of minimum size). But, since $e \in S_\mathrm{min}(X \cup Y)$, we know that $e$ is the leftmost edge from $X \cup Y$ in path $p$. Therefore, $f \not\in X \cup Y$, otherwise neither $X$ nor $Y$ would be cuts. But we know that $f \in S_\mathrm{min}(X \cup Y)$, which means $f \in X \cup Y$, and we reach a contradiction. Thus, the set $S_\mathrm{min}(X \cup Y)$ is a minimum $s$-$t$ cut, and the claim is proven. 
\end{proof}

We now prove Proposition \ref{proposition:1}. We restate it here for the convenience of the reader. 

\propMultCons*
\begin{proof}
We prove this by giving an algorithm that takes any $k$-tuple $C \in U^k$ %of minimum $s$-$t$ cuts 
to a $k$-tuple $\hat{C} \in U^k_{\mathrm{lr}}$ that is in left-right order. The algorithm can be seen in Algorithm \ref{algo:leftRightOrder}. 

\begin{algorithm}[H]
\caption[Caption for LOF]{LRO($C = (X_1, \ldots, X_k)$)} \label{algo:leftRightOrder}
\vspace{.5em}
{\setlist{nolistsep}
    \begin{enumerate}[leftmargin=*, noitemsep]
    \item For each $i \in \{1, \ldots, k-1\}$:
        \begin{enumerate}[align = left]
        \item For each $j \in \{i+1, \ldots, k\}$:
            \begin{enumerate}[align = left]
                \item Set $\hat{X}_r \leftarrow S_{\mathrm{max}}(X_i \cup X_j)$ and $\hat{X}_\ell \leftarrow S_{\mathrm{min}}(X_i \cup X_j)$.
                \item Replace $X_i$ by $\hat{X}_\ell$ and $X_j$ by $\hat{X}_r$.
        \end{enumerate}
    \end{enumerate}
    \item Return $C$.
    \end{enumerate}
}
\vspace{.5em}
\end{algorithm}
We have to verify that for any $k$-tuple $C$, the algorithm %terminates, 
produces a $k$-tuple $\hat{C}$ that is in left-right order, and that $\mu_C(e) = \mu_{\hat{C}}(e)$ for all $e \in E(G)$. %To see that the algorithm terminates just notice that each loop runs for a finite number $k$ of steps. Now, 
To prove 
%see that that $\mu_C(e) = \mu_{\hat{C}}(e)$ for all $e \in E(G)$, 
the latter, 
notice that at iteration $i$ of the algorithm, the two cuts $X_i$ and $X_j$ are replaced by $S_{\mathrm{min}}(X_i \cup X_j)$ and $S_{\mathrm{max}}(X_i \cup X_j)$, respectively. By definition, $S_{\mathrm{min}}(X_i \cup X_j) \cup S_{\mathrm{max}}(X_i \cup X_j) = X_i \cup X_j$ and, by Claim \ref{claim:leftmostRightmost}, we know that $X_i \cap X_j = S_{\mathrm{min}}(X_i \cup X_j) \cap S_{\mathrm{max}}(X_i \cup X_j)$. Therefore, the multiplicity of the edges $e \in E(G)$ remains invariant at every iteration. It then follows that the $k$-tuple $\hat{C} = LRO(C)$ output by the algorithm contains the same set of edges as the input tuple; each of them preserving its multiplicity. 

It remains to show that $\hat{C}$ is in left-right order. First, notice that $\hat{C}$ iterates over every pair of indices $(i, j)$ such that $i < j$. Furthermore, the algorithm sees such a pair only once. Now, assume that $\hat{C}$ is not in left-right order. Then, it contains a pair $(X_i, X_j)$ of incomparable (crossing) cuts; but this cannot be the case, as these would have been replaced by $S_{\mathrm{min}}(X_i \cup X_j)$ and $S_{\mathrm{max}}(X_i \cup X_j)$ at iteration $(i, j)$. Therefore $\hat{C}$ is in left-right order\footnote{Alternatively, one can see that cut $X_i$ at the end of the inner loop satisfies that $X_i \leq X_j$ for all $i < j$; hence, at iteration $i$ of the outer loop the algorithm finds a cut $\hat{X}_i$ to the right of $\hat{X}_{i-1}$ that is leftmost with respect to $\hat{X}_j$ for all $i < j$. That is, $X_{i-1} \leq X_i \leq X_j$ for all $i \in [k]$ and $i < j$.} and the proposition is proved. 
\end{proof}

\subsection{Proof of Lemma \ref{lemma:multiplicityModular}} \label{appendix.multiplicityModular}
Let $C_1 = [X_1, \ldots, X_k]$ and $C_2 = [Y_1, \ldots, Y_k]$ be distinct elements in the lattice $L^* = (U_{\mathrm{lr}}^k, \preceq)$. 
For a fixed edge $e \in E(G)$, we are interested in $\mu_e(C_1 \vee C_2) + \mu_e(C_1 \wedge C_2)$. For this purpose, consider the set of indices $P = \{1, \ldots, k\}$. We partition $P$ into four parts: (i) $P_1 = \{i \, : \, e \notin X_i \cup Y_i\}$, (ii) $P_2 = \{i \, : \, e \in X_i, e \notin Y_i\}$, (iii) $P_3 = \{i \, : \, e \notin X_i, e \in Y_i\}$ and (iv) $P_4 = \{i \, : \, e \in X_i \cap Y_i\}$. 
We claim that $\mu_e(C_1 \vee C_2) + \mu_e(C_1 \wedge C_2) = |P_2| + |P_3| + 2|P_4|$. This follows because on the one hand, by definition, the edge $e$ must appear in either $S_\text{min}(X_i \cup Y_i)$ or $S_\text{max}(X_i \cup Y_i)$ for each $i \in P_2 \cup P_3$. On the other hand, the edge $e$ appears in both $S_\text{min}(X_i \cup Y_i)$ and $S_\text{max}(X_i \cup Y_i)$ for every $i \in P_4$, since there is no edge $f \in X_i \cup Y_i$ on the same $s$-$t$ path $p$ as $e$ such that $f \leq_p e$ or $e \leq_p f$ (otherwise it could not be that $e \in X_i \cap Y_i$). 
Now, observe that from the way we partitioned the set $P$, we have 
$\mu_e(C_1) = |P_2| + |P_4|$ and $\mu_e(C_2) = |P_3| + |P_4|$. Combining this with our previous claim, we obtain $\mu_e(C_1 \vee C_2) + \mu_e(C_1 \wedge C_2) = \mu_e(C_1) + \mu_e(C_2)$. 
By definition of modularity, the multiplicity function $\mu_e$ is thus modular on the lattice $L^*$ for any edge $e \in E(G)$. \hfill$\qed$

\subsection{Proof of Lemma \ref{lemma:multiplicityProperty}} \label{appendix.multiplicityProperty}
We require the following proposition. 

\begin{proposition} \label{proposition.interval}
For any $C = [X_1, \ldots, X_k]$ in $L^*$, the edge $e \in E(C)$ appears in every cut of a contiguous subsequence $C' = [X_i, \ldots, X_j]$ of $C$, $1 \leq i \leq j \leq k$, with size $|C'| = \mu_e(C)$.
\end{proposition}
\begin{proof}
    The case when $\mu_e(C) = 1$ is trivial. Next, we prove the case when $\mu_e(C) \geq 2$. By contradiction, suppose that $e$ does not appear in a contiguous subsequence of $C$. Then, there exists some cut $X_h \in C$ with $i < h < j$ such that $e \in X_i$, $e \not\in X_h$, and $e \in X_j$. We know that collection $C$ is in left-right order, thus we have that $X_i \leq X_j$ for every $i < j$. Now, from $e \in X_i$, it follows that $e$ is a path-predecessor of en edge $f$ in $X_h$. But from $e \in X_j$, edge $e$ must also be a path-successor of $f$. The edges $e$ and $f$ cannot be equal since $e \not\in X_h$, thus we get the necessary contradiction.
\end{proof}

\begin{remark} \label{remark:Interval} \label{appendix.remarkInterval}
By Proposition \ref{proposition.interval}, we can represent the containment of an edge $e$ in a collection $C \in L^*$ as an interval $I_e(C) = (i, j)$, where $i \leq j$, of length $\mu_e(C)$ defined on the set of integers $\{1, \ldots, k\}$. In this interval representation, the elements of $I_e(C)$ correspond bijectively to the positions of the cuts in $C$ that contain edge $e$. This will be useful in the proofs of Lemma \ref{lemma:multiplicityProperty} and Claim \ref{claim:binomialSubmodular}. 
\end{remark}

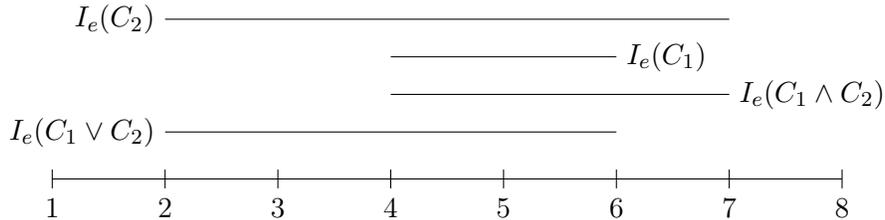
\begin{figure}[H]
    \centering
    \begin{tikzpicture}[scale=0.5]
    \draw (3,0) -- (24,0); %Axis
    \foreach \x in {1,2,...,8} {
        \draw (\x*3,0.25) -- (\x*3,-0.25) node[below] {\x};
    }
    % Intervals
    \draw (6,1.25) node[left] {$I_e(C_1 \vee C_2)$} -- (18,1.25); % I_e(C_1 \vee C_2)
    \draw (12,2.25) -- (21,2.25) node[right] {$I_e(C_1 \wedge C_2)$}; % I_e(C_1 \wedge C_2)
    \draw (12,3.25) -- (18,3.25) node[right] {$I_e(C_1)$}; % I_e(C_1)
    \draw (6,4.25) node[left] {$I_e(C_2)$} -- (21,4.25); % I_e(C_2)   
    
    \end{tikzpicture}
    \caption{%Example illustrating the 
    Interval representation of the containment of an edge $e \in E(G)$ in left-right ordered collections of $s$-$t$ mincuts. For any two $C_1, C_2 \in L^*$ such that $e \in E(C_1) \cup E(C_2)$, we are interested in the containment of $e$ in collections $C_1 \vee C_2$ and $C_1 \wedge C_2$. In the example, there are $k = 8$ cuts in each collection, and 8 corresponding elements in the domain of the intervals. Observe that neither $I_e(C_1 \vee C_2)$ nor $I_e(C_1 \wedge C_2)$ are longer than $I_e(C_1)$ or $I_e(C_2)$. Also, the corresponding sums of their lengths are equal. 
    }
    \label{fig:intervalRepresentation}
\end{figure}

We are now ready to prove Lemma \ref{lemma:multiplicityProperty}. We restate it here for the reader's convenience. 

\multiplicityProperty*
\begin{proof}
We prove this by case distinction on the containment of $e$ in $E(C_1) \cup E(C_2)$. There are three cases: $e \in E(C_1) \setminus E(C_2)$, $e \in E(C_2) \setminus E(C_1)$, and $e \in E(C_1) \cap E(C_2)$. %Cases (1) and (2) are symmetrical so we shall only consider Case 1. 
\begin{description}
   \item[Case 1: $e \in E(C_1) \setminus E(C_2)$.] 
   We prove this case by contradiction. Assume that $\max(\mu_e(C_1 \vee C_2), \mu_e(C_1 \wedge C_2)) > \mu_e(C_1)$. By Lemma \ref{lemma:multiplicityModular}, we know that $\mu_e(C_1 \vee C_2) + \mu_e(C_1 \wedge C_2) = \mu_e(C_1)$. W.l.o.g., we can assume that $\mu_e(C_1 \wedge C_2) > \mu_e(C_1 \vee C_2)$. This implies that $\mu_e(C_1 \vee C_2) < 0$, which is a contradiction. Hence, it must be that $\max({\mu_e(C_1 \vee C_2), \mu_e(C_1 \wedge C_2)}) \leq \mu_e(C_1)$.
   \item[Case 2: $e \in E(C_2) \setminus E(C_1)$.] 
   This case is symmetrical to Case 1, hence is already proven. 
   \item[Case 3: $e \in E(C_1) \cap E(C_2)$.]
   To prove that the statement is true in this case, it is convenient to consider the interval representation of edge $e$ in $E(C_1)$ and $E(C_2)$. Let $I_e(C_1) = (\alpha, \beta)$ and $I_e(C_2) = (\sigma, \tau)$ be such intervals as defined by Remark \ref{remark:Interval}. There are two subcases to consider: $I_e(C_1) \cap I_e(C_2) = \emptyset$, and $I_e(C_1) \cap I_e(C_2) \neq \emptyset$. 
   \begin{description}
       \item[Subcase 3.1.] We claim that $\max({\mu_e(C_1 \vee C_2), \mu_e(C_1 \wedge C_2)}) = \max(\mu_e(C_1), \mu_e(C_2))$ holds in this subcase. To see this, w.l.o.g., suppose that $\beta < \sigma$. Then, because $C_2$ is in left-right order, the cuts of $C_2$ in the interval $(\alpha, \beta)$ each contain a path-predecessor of edge $e$. Then, by definition of the join operation in $L^*$, we have $I_e(C_1 \vee C_2) = (\alpha, \beta)$. Similarly, the cuts of $C_1$ in the interval $(\sigma, \tau)$ each contain a path-successor of $e$. Hence, by the meet operation in $L^*$, we have $I_e(C_1 \wedge C_2) = (\sigma, \tau)$. Taking the length of the intervals, we obtain $\mu_e(C_1 \vee C_2) = \mu_e(C_1)$ and $\mu_e(C_1 \wedge C_2) = \mu_e(C_2)$, from which the claim follows. 
       \item[Subcase 3.2.] We have two further subcases to consider: $I_e(C_1) \not\subseteq I_e(C_2)$ (or $I_e(C_2) \not\subseteq I_e(C_1)$), and $I_e(C_1) \subseteq I_e(C_2)$ (or vice versa). 
       \begin{description}
           \item[Subcase 3.2.1.]
           The proof of this subcase is analogous to the proof of subcase (3.1), where we also obtain that $\max({\mu_e(C_1 \vee C_2), \mu_e(C_1 \wedge C_2)}) = \max(\mu_e(C_1), \mu_e(C_2))$. %This leaves us to consider subcase (3.2.2).
           \item[Subcase 3.2.2.]
           W.l.o.g., suppose that $I_e(C_2) \subseteq I_e(C_1)$ (see Figure \ref{fig:intervalRepresentation} for an illustration). Then $\alpha \leq \sigma \leq \tau \leq \beta$. Again, by definition of join and meet, we have that $I_e(C_1 \vee C_2) = (\alpha, \tau)$ and $I_e(C_1 \wedge C_2) = (\sigma, \beta)$. Now, since $\tau - \alpha \leq \beta - \alpha$ and $\beta - \sigma \leq \beta - \alpha$, we obtain $\max(\mu_e(C_1 \vee C_2), \mu_e(C_1 \wedge C_2)) \leq \max(\mu_e(C_1), \mu_e(C_2))$, which is what we wanted.
       \end{description}
   \end{description}
\end{description}

Since the claim is true for all cases covered and all cases have been considered, the claim is proved.
\end{proof}

\subsection{Proof of Claim \ref{claim:binomialSubmodular}} \label{appendix.binomialSubmodular}
We know that $e \in E(C_1 \vee C_2) \cup E(C_1 \wedge C_2)$ iff $e \in E(C_1) \cup E(C_2)$ (see proof in Appendix \ref{sec.missingProofs}). Hence, we may only consider the edge set $E(C_1) \cup E(C_2)$. We prove the claim by case distinction on the containment of $e$ in $E(C_1) \cup E(C_2)$. There are three cases: $e \in E(C_1) \setminus E(C_2)$, $e \in E(C_2) \setminus E(C_1)$, and $e \in E(C_1) \cup E(C_2)$.
\begin{description}
    \item[Case 1: $e \in E(C_1) \setminus E(C_2)$.]
    We know from Lemma \ref{lemma:multiplicityProperty} that $\mu_e(C_1 \vee C_2) \leq \mu_e(C_1)$ and $\mu_e(C_1 \wedge C_2) \leq \mu_e(C_1)$. Hence we have $\binom{\mu_e(C_1 \vee C_2)}{2} \leq \binom{\mu_e(C_1)}{2}$ and $\binom{\mu_e(C_1 \wedge C_2)}{2} \leq \binom{\mu_e(C_1)}{2}$. Moreover, from Lemma \ref{lemma:multiplicityModular}, we know that $\mu_e(C_1 \vee C_2) + \mu_e(C_1 \wedge C_2) = \mu_e(C_1)$. It is clear that $\binom{a}{2} + \binom{b}{2} < \binom{a + b}{2}$ for any $a, b \in \mathbb{N}$. Therefore, the claim is satisfied in this case.
    \item[Case 2: $e \in E(C_2) \setminus E(C_1)$.]
    This case is symmetrical to Case 1, hence is already proven.
    \item[Case 3: $e \in E(C_1) \cup E(C_2)$.]
    Consider the interval representation of $e$ in $E(C_1) \cup E(C_2)$ (see Remark \ref{remark:Interval} for details). There are three subcases: (3.1) $I_e(C_1)$ and $I_e(C_2)$ have no overlap (i.e., $I_e(C_1) \cap I_e(C_2) = \emptyset$), (3.2) $I_e(C_1)$ and $I_2(C_2)$ overlap but neither is entirely contained in the other (i.e., $I_e(C_1) \cap I_e(C_2) \neq \emptyset$ and $I_e(C_1) \not\subseteq I_e(C_2)$ nor $I_e(C_2) \not\subseteq I_e(C_1)$), and (3.3) one of $I_e(C_1)$ or $I_e(C_2)$ is entirely contained in the other (i.e., $I_e(C_1) \subseteq I_e(C_2)$ or $I_e(C_2) \subseteq I_e(C_1)$). 
    \begin{description}
        \item[Subcase 3.1.]
        We know by the proof of Lemma \ref{lemma:multiplicityProperty} that $\max({\mu_e(C_1 \vee C_2), \mu_e(C_1 \wedge C_2)}) = \max(\mu_e(C_1), \mu_e(C_2))$. And by Lemma \ref{lemma:multiplicityModular}, we also have $\min({\mu_e(C_1 \vee C_2), \mu_e(C_1 \wedge C_2)}) = \min(\mu_e(C_1), \mu_e(C_2))$. It is then immediate that the claim is satisfied with equality in this case. 
        
        \item[Subcase 3.2.]
        Analogous to Subcase 3.1.
        
        \item[Subcase 3.3.]
        It is easy to show that $\binom{a}{2} + \binom{b}{2} \leq \binom{c}{2} + \binom{d}{2}$ for $a, b, c, d \in \mathbb{N}$, given that the following properties hold: $a + b = c + d$, and $\max(a, b) \leq \max(c, d)$.\footnote{By combining the two properties, we have $a \cdot b \geq c \cdot d$. Moreover, by the former property, we know that $(a+b)^2 = (c+d)^2$. Together, these facts imply that $a^2 + b^2 \leq c^2 + d^2$. Again by the first property, we can subtract $(a+b)$ and $(c+d)$ from each side, respectively, resulting in $a(a-1) + b(b-1) \leq c(c-1) + d(d-1)$. Then, by definition of the binomial coefficient, the claim immediately follows.} In our context, these are the properties satisfied by the multiplicity function stated in Lemmas \ref{lemma:multiplicityModular} and \ref{lemma:multiplicityProperty}. Therefore, the claim is also satisfied in this subcase.
    \end{description}
\end{description}

Since we have considered all cases, we conclude that the claim holds for every edge $e \in E(C_2) \cup E(C_1)$. And, since it holds vacuously for edges $e \not\in E(C_2) \cup E(C_1)$, the generalized claim follows. \hfill$\qed$

\subsection{Proof of Claim \ref{claim:edgeSetSizesInequality}} \label{appendix.edgeSetSizesInequality}
To prove this claim, we shall look at each edge in the graph, and inspect whether it is contained in each of the four edge sets $E(C_1 \vee C_2)$, $E(C_1 \wedge C_2)$, $E(C_1)$, and $E(C_2)$. 
For simplicity, we shall denote these sets by $A$, $B$, $C$, and $D$, respectively. 
We begin with two simple facts. First, for any two sets $X$ and $Y$, we may write $|X|+|Y| = |X \setminus Y| + |Y \setminus X| + 2 |X \cap Y|$. Second, any edge $e \in A \cup B$ iff $e \in C \cup D$ (see proof in Appendix \ref{sec.missingProofs}). 
Thus, for an edge $e$, we can restrict to analyze the following cases: $e \in A \setminus B$, $e \in B \setminus A$, and $e \in A \cap B$. 
%Cases (i) and (ii) are symmetrical, so we shall only consider (i). 

\begin{description}
    \item[Case 1: $e \in A \setminus B$.]
    We claim that $e$ must be contained in either $C$ or $D$, but not in both. By contradiction, assume that $e$ is contained in $C \cap D$. Then, by Lemma \ref{lemma:multiplicityModular}, we have $\mu_e(A) = \mu_e(C) + \mu_e(D)$. But this implies that $\mu_e(A) > \max(\mu_e(C), \mu_e(D))$, which stands in contradiction with Lemma \ref{lemma:multiplicityProperty}.  
    Therefore, every time an edge appears exclusively in $A$ or $B$, it also appears exclusively in $C$ or $D$ (observe that the reverse is not always true). More formally, we have $|C \setminus D| + |D \setminus C| = |A \setminus B| + |B \setminus A| + |X|$, where $X$ is the set of edges in $C$ or $D$ that could also appear in $A \cap B$. 
    \item[Case 2: $e \in B \setminus A$.]
    Symmetrical to Case 1. 
    \item[Case 3: $e \in A \cap B$.]
    We have three subcases: (3.1) $e \in C \setminus D$, (3.2) $e \in D \setminus C$, and (3.3) $e \in C \cap D$. 
    \begin{description}
        \item[Subcases 3.1 \& 3.2.]
        An observant reader may notice that these subcases are equivalent to inspecting the set $X$. In fact, it is enough for our purposes to show that $|X| \geq 0$. To see why this holds, assume w.l.o.g. that $e \in C \setminus D$ in the interval $(i, j)$. Now, consider the case where $D$ contains only cuts that each contains a path-predecessor of $e$ in the interval $(1, h)$ and cuts that each contains a path-successor of $e$ in the interval $(h+1, k)$, with $i < h < j$. Then, by definition of join and meet in $L^*$, $e \in A$ in the interval $(1,h)$ and $e \in B$ in the interval $(h+1, k)$. Therefore, $e \in A \cap B$, which implies $|X| \geq 0$.
        \item[Subcase 3.3.]
        By a similar argument to Case 1, it follows that every edge that appears in $C \cap D$ also appears in $A \cap B$. Therefore, $|A \cap B| = |C \cap D| + |X|$.
    \end{description}
\end{description}

Putting everything together, we obtain the following:
\begin{align*}
    |C| + |D| & = |C \setminus D| + |D \setminus C| + 2 |C \cap D| \notag \\
    & = |A \setminus B| + |B \setminus A| + |X| + 2 |A \cap B| - 2 |X| \notag \\
    & = |A| + |B| - |X|, 
\end{align*}
and since $|X| \geq 0$, we have that $|A| + |B| \geq |C| + |D|$ and the claim is proven. \hfill$\qed$

\section{Solving \textsc{SUM-k-DMC} and \textsc{COV-k-DMC} in \texorpdfstring{$O(k^5n^5)$}{O(k5n5)} time} \label{appendix.runningTime}
\preciseruntime*
\begin{proof}
Let us first look at each step in the reduction to SFM. From the discussion in Section \ref{sec.sfmPreliminaries}, to minimize a submodular function $f$ in a distributive lattice $L$, we should first transform the problem into an equivalent version on sets. For this, we require (i) a compact representation of the lattice $L$, and (ii) a polynomial evaluation oracle for the function $\hat{f} : \mathcal{D}(J(L)) \rightarrow \mathbb{R}$, which is an analogue of $f$ on $L$ but defined on the poset of ideals of the join-irreducibles of $L$. Then, any algorithm for SFM on sets can be used to solve the original problem. 

In our context, the total running time of the algorithm is, 
\begin{center}
    $O(t_\mathrm{c}(n, m) + t_\mathrm{SFM}(n, m, T_\mathrm{EO}))$,
\end{center}
where $t_\mathrm{c}(n, m)$ is the time required to construct a compact representation of the lattice $L^*$, $t_\mathrm{SFM}(n, m, T_\mathrm{EO})$ is the time taken by an SFM algorithm on sets, and $T_\mathrm{EO}$ is the time required to compute $\hat{d}'_\mathrm{sum}$ (resp. $\hat{d}'_\mathrm{cov}$) by an evaluation oracle. Here $\hat{d}'_\mathrm{sum}$ (resp. $\hat{d}'_\mathrm{cov}$) denotes the analog function on $\mathcal{D}(J(L^*))$ of $\hat{d}_\mathrm{sum}$ (resp. $\hat{d}_\mathrm{cov}$) on $L^*$, and $n$ and $m$ denote the number of vertices and edges of our input graph $G$. 

By Lemmas \ref{lemma:picard} and \ref{lemma:joinIrreduciblesLStar}, a compact representation of $L^*$ can be computed in time $t_\mathrm{c}(n, m) = F(n, m) + O(k^2n^2)$, where $F(n, m)$ is the time required by a max-flow computation. The latter term follows from considering all potential edges between pairs of nodes in the construction of $G(L^*)$ from the set of join-irreducibles $J(L^*)$, which has size $O(kn)$. 

Next, we analyze the time $T_\mathrm{EO}$ required to compute $\hat{d}'_\mathrm{sum}(A)$; where $A$ is the set of join-irreducibles lower than or equal to the corresponding element $a \in L^*$. We know that the original function $\hat{d}_\mathrm{sum}(a)$ can be efficiently computed in $O(kn)$ time. Hence, if we can recover the element $a \in L^*$ from the ideal $A \in \mathcal{D}(J(L))$ in $t_\mathrm{ideal}(n, m)$ time, we can use our familiar function $\hat{d}_\mathrm{sum}$ to compute $\hat{d}'_\mathrm{sum}$ in time $t_\mathrm{ideal}(n, m) + O(kn)$. We claim that $t_\mathrm{ideal}(n, m) = O(k^2n^2)$. This follows from the fact that $a$ can be recovered from $A$ by computing the join of all $O(kn)$ elements in $A$, where a join operation between two elements in $L^*$ has complexity $O(k \lambda(G)) = O(kn)$. 

Using the above, and plugging into $t_\mathrm{SFM}(n, m, T_\mathrm{EO})$ the running time of the current best SFM algorithm of Jiang \cite{jiang2021minimizing}, we obtain a total running time of $F(n, m) + O(k^2n^2) + O(k^3 n^3 \cdot k^2n^2)$, which simplifies to $O(k^5n^5)$. 
\end{proof}

\section{Other proofs}\label{sec.missingProofs}
\subsection{Fact used in the proofs of Claims \ref{claim:binomialSubmodular} and \ref{claim:edgeSetSizesInequality}}

\begin{fact}
For any two $C_1, C_2 \in L^*$, we have that $e \in E(C_1 \vee C_2) \cup E(C_1 \wedge C_2)$ iff $e \in E(C_1) \cup E(C_2)$.
\end{fact}
\begin{proof}
    Assume, for the sake of contradiction, that there exists an edge $e$ such that $e \in E(C_1) \cup E(C_2)$ but $e \not\in E(C_1 \vee C_2) \cup E(C_1 \wedge C_2)$. If $e \in E(C_1) \cup E(C_2)$, then there exists at least one cut $X_i \in C_1 \cup C_2$ such that $e \in X_i$. W.l.o.g., suppose that $X_i \in C_1$, and let $Y_i$ be the $i$th cut in $C_2$. If $e \not\in S_\mathrm{min}(X_i \cup Y_i)$ then, by definition, it must be that $e \in S_\mathrm{max}(X_i \cup Y_i)$ and vice versa. %But we had assumed that $e \not\in E(C_1 \vee C_2) \cup E(C_1 \wedge C_2)$. Hence, we obtain our 
    This gives us the necessary contradiction and $e$ must also be contained in $E(C_1 \vee C_2) \cup E(C_1 \wedge C_2)$. The other direction of the argument is similar and is thus omitted. 
\end{proof}
\end{document}